%% file: mainarxiv.tex
\documentclass[]{article}

\usepackage{fullpage}
\usepackage{amsthm}

\newtheorem{proposition}[]{Proposition}
\newtheorem{example}[]{Example}
\newtheorem{remark}[]{Remark}
\newtheorem{theorem}[]{Theorem}
\newtheorem{lemma}[]{Lemma}
\newtheorem{corollary}[]{Corollary}

\input{preamble}
\begin{document}

\title{Parikh Automata over Infinite Words\thanks{Shibashis Guha is supported by the DST-SERB project SRG/2021/000466 \emph{Zero-sum and Nonzero-sum Games for Controller Synthesis of Reactive Systems}. Ismaël Jecker is supported by the ERC grant 950398 (INFSYS). Martin Zimmermann is supported by DIREC – Digital Research Centre Denmark.}}
\newcommand{\email}[1]{\texttt{#1}}
\date{\vspace{-.3cm}}
\author{\small Shibashis Guha\\
\small Tata Institute of Fundamental Research, Mumbai, India\\\medskip 
\small \email{shibashis.guha@tifr.res.in}\\ 
\small Ismaël Jecker\\
\small University of Warsaw, Poland \\\medskip
\small \email{ismael.jecker@gmail.com}\\
\small Karoliina Lehtinen\\
\small CNRS, Aix-Marseille University and University of Toulon, LIS,
\small Marseille, France\\\medskip
\small \email{lehtinen@lis-lab.fr}\\
\small Martin Zimmermann\\
\small Aalborg University of Aalborg, Denmark\\
 \small \email{mzi@cs.aau.dk}
}
\maketitle
\begin{abstract}
\input{arxiv_abstract}
\end{abstract}

\input{arxiv_content}

\end{document}

%% file: preamble.tex
\usepackage[utf8]{inputenc}

\usepackage{amsmath}
\usepackage{amssymb}
\usepackage{amsthm}

\usepackage{xspace}
\usepackage{booktabs}
\usepackage{colortbl}

\usepackage{pifont}

\usepackage{tikz}
\usetikzlibrary{calc,matrix,positioning}
\usetikzlibrary{decorations.pathreplacing}
\usetikzlibrary{automata}
\tikzset{>=stealth, shorten >=1pt}
\tikzset{every edge/.style = {thick, ->, draw}}
\tikzset{every loop/.style = {thick, ->, draw}}

\newcommand{\myquot}[1]{``#1''}

\newcommand{\yess}{\ding{51}}%
\newcommand{\noo}{\ding{55}}%
\newcommand{\quest}{\bf ?}

\newcommand{\nats}{\mathbb{N}}
\newcommand{\set}[1]{\{#1\}}

\newcommand{\ind}[1]{I_{#1}}
\newcommand{\size}[1]{|#1|}

\renewcommand{\epsilon}{\varepsilon}
\newcommand{\occ}[2]{|#1|_{#2}}

\newcommand{\NP}[0]{NP\xspace}
\newcommand{\coNP}[0]{coNP\xspace}
\newcommand{\Lall}{L_\text{S}}
\newcommand{\Lsexists}{L_\text{R}^s}
\newcommand{\Laexists}{L_{\text{R}}^a}
\newcommand{\LsBuchi}{L_{\text{B}}^s}
\newcommand{\LaBuchi}{L_{\text{B}}^a}
\newcommand{\LcoBuchi}{L_{\text{C}}}

\newcommand{\Lcal}{\mathcal{L}}

\newcommand{\doublebal}{D}
\newcommand{\bal}{B}
\newcommand{\unbal}{U}
\newcommand{\Fina}{A_{<\omega}}
\newcommand{\Infa}{A_{=\omega}}
\newcommand{\Exa}{A_{>0}}
\newcommand{\almostallunbal}{U'}

\newcommand{\aut}{\mathcal{A}}

\newcommand{\trans}{\mathcal{T}}

\newcommand{\init}{I}
\newcommand{\run}{\rho}

\newcommand{\projsigma}{p_\Sigma}
\newcommand{\parikhimage}{\Phi_e}

\newcommand{\pa}{PA\xspace}
\newcommand{\nfa}{NFA\xspace}
\newcommand{\dfa}{DFA\xspace}

\newcommand{\rpa}{RPA\xspace}
\newcommand{\srpa}{sRPA\xspace}
\newcommand{\arpa}{aRPA\xspace}
\newcommand{\spa}{SPA\xspace}

\newcommand{\sbpa}{sBPA\xspace}
\newcommand{\bpa}{BPA\xspace}
\newcommand{\abpa}{aBPA\xspace}
\newcommand{\cpa}{CPA\xspace}

\newcommand{\mach}{\mathcal{M}}
\newcommand{\stopp}{\texttt{STOP}}
\newcommand{\instr}{\texttt{I}}
\newcommand{\inc}[1]{\texttt{INC(X}_{#1}\texttt{)}}
\newcommand{\dec}[1]{\texttt{DEC(X}_{#1}\texttt{)}}
\newcommand{\ite}[3]{\texttt{IF X}_{#1}\texttt{=0 GOTO }#2\texttt{ ELSE GOTO }#3}

\newcommand{\balance}[1]{\mathsf{diff}(#1)}
\newcommand{\balanceprime}[1]{\mathsf{diff}'(#1)}
\newcommand{\SET}[1]{[#1]}

\newcommand{\greekC}{\pi_c}
\newcommand{\greekABA}{\chi}
\newcommand{\greekD}{\pi_d}
\newcommand{\greekA}{\pi_a}
\newcommand{\greekB}{\pi_b}

%% file: arxiv_abstract.tex
Parikh automata extend finite automata by counters that can be tested for membership in a semilinear set, but only at the end of a run, thereby preserving many of the desirable algorithmic properties of finite automata. 
Here, we study the extension of the classical framework onto infinite inputs: We introduce reachability, safety, Büchi, and co-Büchi Parikh automata on infinite words and study expressiveness, closure properties, and the complexity of verification problems.

We show that almost all classes of automata have pairwise incomparable expressiveness, both in the deterministic and the nondeterministic case; a result that sharply contrasts with the well-known hierarchy in the $\omega$-regular setting.
Furthermore, emptiness is shown decidable for Parikh automata with reachability or Büchi acceptance, but undecidable for safety and co-Büchi acceptance.
Most importantly, we show decidability of model checking with specifications given by deterministic Parikh automata with safety or co-Büchi acceptance, but also undecidability for all other types of automata.
Finally, solving games is undecidable for all types.

%% file: arxiv_content.tex
\section{Introduction}
\label{sec_intro}
While finite-state automata are the keystone of automata-theoretic verification, they are not expressive enough to deal with the many nonregular aspects of realistic verification problems. Various extensions of finite automata have emerged over the years, to allow for the specification of context-free properties and beyond, as well as the modelling of timed and quantitative aspects of systems. Among these extensions, Parikh automata, introduced by Klaedtke and Rueß~\cite{KR}, consist of finite automata augmented with counters that can only be incremented. A Parikh automaton only accepts a word if the final counter-configuration is within a semilinear set specified by the automaton. As the counters do not interfere with the control flow of the automaton, that is, counter values do not affect whether transitions are enabled, they allow for mild quantitative computations without the full power of vector addition systems or other more powerful models. 

For example, the nonregular language of words that have more $a$'s than $b$'s is accepted by a Parikh automaton obtained from the one-state \dfa accepting $\set{a,b}^*$ by equipping it with two counters, one counting the $a$'s in the input, the other counting the $b$'s, and a semilinear set ensuring that the first counter is larger than the second one.
With a similar approach, one can construct a Parikh automaton accepting the non-context-free language of words that have more $a$'s than $b$'s and more $a$'s than $c$'s.

Klaedtke and Rueß~\cite{KR} showed Parikh automata to be expressively equivalent to a quantitative version of existential WMSO that allows for reasoning about set cardinalities. Their expressiveness also coincides with that of reversal-bounded counter machines~\cite{KR}, in which counters can go from decrementing to incrementing only a bounded number of times, but in which counters affect control flow~\cite{Ibarra78}. The (weakly) unambiguous restriction of Parikh automata, that is, those that have at most one accepting run, on the other hand, coincide with unambiguous reversal-bounded counter machines~\cite{BCKN20}. Parikh automata are also expressively equivalent to weighted finite automata over the groups $(\mathbb{Z}^k, +, 0)$~\cite{DM00,MS01} for $k \geqslant 1$. This shows that Parikh automata accept a natural class of quantitative specifications.

Despite their expressiveness, Parikh automata retain some decidability: nonemptiness, in particular, is \NP-complete~\cite{FL}. For weakly unambiguous Parikh automata, inclusion~\cite{CM17} and regular separability~\cite{CCLP17} are decidable as well.
Figueira and Libkin~\cite{FL} also argued that this model is well-suited for querying graph databases, while mitigating some of the complexity issues related with more expressive query languages. Further, they have been used in the model checking of transducer properties~\cite{FMR20}.
 
As Parikh automata have been established as a robust and useful model, many variants thereof exist: pushdown (visibly~\cite{DFT19} and otherwise~\cite{Kar04}), two-way with~\cite{DFT19} and without stack~\cite{FGM19}, unambiguous~\cite{CFM13}, and weakly unambiguous~\cite{BCKN20} Parikh automata, to name a few. 
Despite this attention, so far, some more elementary questions have remained unanswered. For instance, despite Klaedtke and Rueß's suggestion in~\cite{KR} that the model could be extended to infinite words, we are not aware of previous work on $\omega$-Parikh automata.

Yet, specifications over infinite words are a crucial part of the modern verification landscape. Indeed, programs, especially safety-critical ones, are often expected to run continuously, possibly in interaction with an environment. Then, executions are better described by infinite words, and accordingly, automata over infinite, rather than finite, words are appropriate for capturing specifications.

This is the starting point of our contribution: we extend Parikh automata to infinite inputs, and consider reachability, safety, Büchi, and co-Büchi acceptance conditions. We observe that when it comes to reachability and Büchi, there are two possible definitions: an asynchronous one that just requires both an accepting state and the semilinear set to be reached (once or infinitely often) by the run, but not necessarily at the same time, and a synchronous one that requires both to be reached (once or infinitely often) simultaneously. 
Parikh automata on infinite words accept, for example,  the languages of infinite words 
\begin{itemize}
    \item with some prefix having more $a$'s than $b$'s (reachability acceptance), 
    \item with all nonempty prefixes having more $a$'s than $b$'s (safety acceptance),  
    \item with infinitely many prefixes having more $a$'s than $b$'s (Büchi acceptance), and
    \item with almost all prefixes having more $a$'s than $b$'s (co-Büchi acceptance).
\end{itemize}

We establish that, both for reachability and Büchi acceptance, both the synchronous and the asynchronous variant are linearly equivalent in the presence of nondeterminism, but not for deterministic automata.
Hence, by considering all acceptance conditions and (non)determinism, we end up with twelve different classes of automata.
We show that almost all of these classes have pairwise incomparable expressiveness, which is in sharp contrast to the well-known hierarchies in the $\omega$-regular case. 
Furthermore, we establish an almost complete picture of the Boolean closure properties of these twelve classes of automata. 
Most notably, they lack closure under negation, even for nondeterministic Büchi Parikh automata.
Again, this result should be contrasted with the $\omega$-regular case, where nondeterministic Büchi automata are closed under negation~\cite{Buchi}.

We then study the complexity of the most important verification problems, e.g., nonemptiness, universality, model checking, and solving games.
We show that nonemptiness is undecidable for deterministic safety and co-Büchi Parikh automata. However, perhaps surprisingly, we also show that nonemptiness is decidable, in fact \NP-complete, for reachability and Büchi Parikh automata, both for the synchronous and the asynchronous versions. 
Strikingly, for Parikh automata, the Büchi acceptance condition is algorithmically simpler than the safety one (recall that their expressiveness is pairwise incomparable).

Next, we consider model checking, arguably the most successful application of automata theory in the field of automated verification. 
Model checking asks whether a given finite-state system satisfies a given specification. 
Here, we consider quantitative specifications given by Parikh automata. 
Model checking is decidable for specifications given by deterministic Parikh automata with safety or co-Büchi acceptance. 
On the other hand, the problem is undecidable for all other classes of automata.

The positive results imply that one can model-check an arbiter serving requests from two clients against specifications like \myquot{the accumulated waiting time between requests and responses of client~$1$ is always at most twice the accumulated waiting time for client~$2$ and vice versa} and \myquot{the difference between the number of responses for client~$1$ and the number of responses for client~$2$ is from some point onward bounded by $100$}.
Note that both properties are not $\omega$-regular.

Finally, we consider solving games with winning conditions expressed by Parikh automata. Zero-sum two-player games are a key formalism used to model the interaction of programs with an uncontrollable environment. In particular, they are at the heart of solving synthesis problems in which, rather than verifying the correctness of an existing program, we are interested in generating a program that is correct by construction, from its specifications. In these games, the specification corresponds to the winning condition: one player tries to build a word (i.e., behaviour) that is in the specification, while the other tries to prevent this. As with model checking, using Parikh automata to capture the specification would enable these well-understood game-based techniques to be extended to mildly quantitative specifications. However, we show that games with winning conditions specified by Parikh automata are undecidable for all acceptance conditions we consider.

All proofs omitted due to space restrictions can be found in the appendix.

\section{Definitions}
\label{sec_definitions}

An alphabet is a finite nonempty set~$\Sigma$ of letters. As usual, $\epsilon$ denotes the empty word, $\Sigma^*$ ($\Sigma^+$, $\Sigma^\omega$) denotes the set of finite (finite nonempty, infinite) words over $\Sigma$.
The length of a finite word~$w$ is denoted by $\size{w}$ and, for notational convenience, we define $\size{w} = \infty$ for infinite words~$w$.

The number of occurrences of the letter~$a$ in a finite word~$w$ is denoted by $\occ{w}{a}$.
Let $a,b \in \Sigma$. A word~$w \in \Sigma^*$ is $(a,b)$-balanced if $\occ{w}{a} = \occ{w}{b}$, otherwise it is $(a,b)$-unbalanced.
Note that the empty word is $(a,b)$-balanced.

\subparagraph*{Semilinear Sets}
Let $\nats$ denote the set of nonnegative integers. 
Let $\vec{v} = (v_0, \ldots, v_{d-1}) \in \nats^d$ and $\vec{v}\,' = (v'_{0}, \ldots, v'_{d'-1}) \in \nats^{d'}$ be a pair of vectors. We define their concatenation as $\vec{v} \cdot \vec{v}\,' = (v_0, \ldots, v_{d-1}, v'_{0}, \ldots, v'_{d'-1}) \in \nats^{d+d'}$.
We lift the concatenation of vectors to sets $D \subseteq \nats^d$ and $D' \subseteq \nats^{d'}$via $D\cdot D' = \set{\vec{v} \cdot \vec{v}\,' \mid \vec{v} \in D \text{ and } \vec{v}\,' \in D'}$.

Let $d \geqslant 1$. A set~$C \subseteq \nats^d$ is \emph{linear} if there are vectors~$\vec{v}_0, \ldots, \vec{v}_k \in \nats^d$ such that 
\[
C = \left\{ \vec{v}_0 + \sum\nolimits_{i=1}^k c_i\vec{v}_i \:\middle|\: c_i \in \nats \text{ for } i=1,\ldots, k \right\}.
\]
Furthermore, a subset of $\nats^d$ is \emph{semilinear} if it is a finite union of linear sets. 

\begin{proposition}[\cite{GS}]\label{propsemilinearclosure}
If $C,C' \subseteq \nats^d$ are semilinear, then so are $C\cup C'$, $C\cap C'$, $\nats^d \setminus C$, as well as $\nats^{d'} \cdot C$ and $C \cdot \nats^{d'}$ for every $d' \geqslant 1$. 
\end{proposition}

\subparagraph*{Finite Automata}

A (nondeterministic) finite automaton (\nfa)~$\aut = (Q,\Sigma, q_\init, \Delta, F)$ over $\Sigma$ consists of a finite set~$Q$ of states containing the initial state~$q_\init$, an alphabet~$\Sigma$, a transition relation~$\Delta \subseteq Q \times \Sigma \times Q$, and a set~$F \subseteq Q$ of accepting states.
The \nfa is deterministic (i.e.,~a \dfa) if for every state~$q \in Q$ and every letter~$a \in \Sigma$, there is at most one $q'\in Q$ such that $(q,a,q')$ is a transition of $\aut$.
A run of $\aut$ is a (possibly empty) sequence~$(q_0, w_0, q_1) (q_1, w_1, q_2) \cdots (q_{n-1}, w_{n-1}, q_n) $ of transitions with $q_0 = q_\init$. 
It processes the word~$w_0 w_1 \cdots w_{n-1} \in \Sigma^*$. The run is \emph{accepting} if it is either empty and the initial state is accepting or if it is nonempty and $q_n$ is accepting. 
The language~$L(\aut)$ of $\aut$ contains all finite words $w \in \Sigma^*$ such that $\aut$ has an accepting run processing $w$.

\subparagraph*{Parikh Automata}
Let $\Sigma$ be an alphabet, $d \geqslant 1$, and $D$ a finite subset of $\nats^d$. 
Furthermore, let $w = (a_0,\vec{v}_0) \cdots (a_{n-1},\vec{v}_{n-1})$ be a word over $\Sigma \times D$. 
The $\Sigma$-projection of $w$ is $\projsigma(w) = a_0 \cdots a_{n-1} \in\Sigma^*$ and its \emph{extended Parikh image} is $\parikhimage(w)=\sum_{j=0}^{n-1} \vec{v}_j \in \nats^d$ with the convention~$\parikhimage(\epsilon) = \vec{0}$, where $\vec{0}$ is the $d$-dimensional zero vector.

A \emph{Parikh automaton} (\pa) is a pair $(\aut, C)$ such that $\aut$ is an \nfa over $\Sigma \times D$ for some input alphabet~$\Sigma$ and some finite $D \subseteq \nats^d$ for some $d \geqslant 1$, and $C \subseteq \nats^d$ is semilinear. 
The language of $(\aut, C)$ consists of the $\Sigma$-projections of words~$w \in L(\aut)$ whose extended Parikh image is in $C$, i.e.,
\[
L(\aut, C) = \set{ \projsigma(w) \mid w \in L(\aut) \text{ with } \parikhimage(w)\in C }.
\]
The automaton~$(\aut, C)$ is deterministic, if for every state~$q$ of $\aut$ and every $a \in \Sigma$, there is at most one pair~$(q', \vec{v}) \in Q \times D$ such that $(q,(a,\vec{v}),q')$ is a transition of $\aut$. Note that this definition does \emph{not} coincide with $\aut$ being deterministic:
As mentioned above, $\aut$ accepts words over $\Sigma\times D$ while $(\aut, C)$ accepts words over $\Sigma$.
Therefore, determinism is defined with respect to $\Sigma$ only. 

Note that the above definition of $L(\aut, C)$ coincides with the following alternative definition via accepting runs:
A run $\run$ of $(\aut, C)$ is a run \[\rho = (q_0, (a_0, \vec{v}_0), q_1) (q_1, (a_1, \vec{v}_1), q_2) \cdots (q_{n-1}, (a_{n-1}, \vec{v}_{n-1}), q_{n})\] of $\aut$. 
We say that $\rho$ \emph{processes} the word~$a_0 a_1 \cdots a_{n-1}\in \Sigma^*$, i.e., the $\vec{v}_j$ are ignored, and that $\rho$'s extended Parikh image is $\sum_{j=0}^{n-1} \vec{v}_j$.
The run is accepting, if it is either empty and both the initial state of $\aut$ is accepting and the zero vector (the extended Parikh image of the empty run) is in $C$, or if it is nonempty, $q_n$ is accepting, and $\rho$'s extended Parikh image is in $C$.
Finally, $(\aut, C)$ accepts $w \in \Sigma^*$ if it has an accepting run processing $w$.

\begin{example}\label{example_intro}
Consider the deterministic \pa~$(\aut, C)$ with $\aut$ in Figure~\ref{fig_example_intro} and $
C = \set{(n,n) \mid n\in\nats} \cup \set{(n,2n) \mid n \in\nats}$. 
It accepts the language~$\set{a^nb^n \mid n\in\nats} \cup \set{a^nb^{2n} \mid n \in\nats}$.
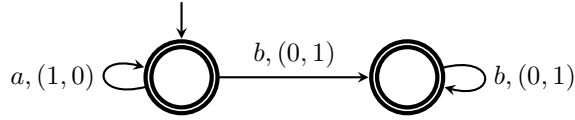
\begin{figure}
    \centering
    \begin{tikzpicture}[ultra thick]
    
    \node[state,accepting] (1) at (0,0) {};
    \node[state,accepting] (2) at (3,0) {};
    
    \path[-stealth]
    (0,1) edge (1)
    (1) edge[loop left] node[left] {$a,(1,0)$} ()
    (1) edge node[above] {$b,(0,1)$} (2)
    (2) edge[loop right] node[right] {$b,(0,1)$} ();
    
    \end{tikzpicture}
    \caption{The automaton for Example~\ref{example_intro}.}
    \label{fig_example_intro}
\end{figure}
\end{example}

A cycle is a nonempty finite run infix
\[(q_0, w_0, q_1) (q_1, w_1, q_2) \cdots (q_{n-1}, w_{n-1}, q_n)(q_n, w_n, q_0)\]
starting and ending in the same state and such that the $q_j$ are pairwise different. 
Note that every run infix containing at least $n$ transitions contains a cycle, where $n$ is the number of states of the automaton.
Many of our proofs rely on the following shifting argument, which has been used before to establish inexpressibility results for Parikh automata~\cite{CFM11}.

\begin{remark}
\label{remark:shifting}
Let $\rho_0 \rho_1 \rho_2 \rho_3$ be a run of a \pa such that $\rho_1$ and $\rho_3$ are cycles starting in the same state.
Then, $\parikhimage(\rho_0 \rho_1 \rho_2 \rho_3) = \parikhimage(\rho_0 \rho_2 \rho_1\rho_3) = \parikhimage(\rho_0 \rho_1 \rho_3 \rho_2 )$.
Furthermore, all three runs end in the same state and visit the same set of states (but maybe in different orders).
\end{remark}


\section{Parikh Automata over Infinite Words}
\label{sec_infinitewords}

In this section, we introduce Parikh automata over infinite words by lifting safety, reachability, Büchi, and co-Büchi acceptance from finite automata to Parikh automata.
Recall that a Parikh automaton on finite words accepts if the last state of the run is accepting and the extended Parikh image of the run is in the semilinear set, i.e., both events are synchronized.
For reachability and Büchi acceptance it is natural to consider both a synchronous and an asynchronous variant while for safety and co-Büchi there is only a synchronous variant. 

All these automata have the same format as Parikh automata on finite words, but are now processing infinite words.
Formally, consider $(\aut, C)$ with $\aut = (Q, \Sigma \times D, q_\init, \Delta, F)$.
Fix an infinite run~$(q_0, w_0, q_1)(q_1, w_1, q_2)(q_2, w_2, q_3) \cdots$ of $\aut$ with $q_0 = q_\init$ (recall that each $w_j$ is in $\Sigma \times D$), which we say \emph{processes} $\projsigma(w_0w_1w_2\cdots)$.
\begin{itemize}

    \item The run is \emph{safety accepting} if $\parikhimage(w_0\cdots w_{n-1}) \in C$ and $q_n \in F$  for all $n\geqslant 0$.
    
    \item The run is \emph{synchronous reachability accepting} if $\parikhimage(w_0\cdots w_{n-1}) \in C$ and $q_n \in F$ for some $n \geqslant 0$.
    
    \item The run is \emph{asynchronous reachability accepting} if $\parikhimage(w_0\cdots w_{n-1}) \in C$ for some $n \geqslant 0$ and $q_{n'} \in F$ for some $n'\geqslant 0$.
     
     \item The run is \emph{synchronous Büchi accepting} if $\parikhimage(w_0 \cdots w_{n-1}) \in C$ and $q_n \in F$ for infinitely many $n \geqslant 0$.
    
    \item The run is \emph{asynchronous Büchi accepting} if $\parikhimage(w_0 \cdots w_{n-1}) \in C$ for infinitely many $n \geqslant 0$ and $q_{n'} \in F$ for infinitely many $n' \geqslant 0$.
    
    \item The run is \emph{co-Büchi accepting} if there is an $n_0$ such that $\parikhimage(w_0 \cdots w_{n-1}) \in C$ and $q_n \in F$ for every $n \geqslant n_0$.
    
\end{itemize}
As mentioned before, we do not distinguish between synchronous and asynchronous co-Büchi acceptance, as these definitions are equivalent.
Also, note that all our definitions are conjunctive in the sense that acceptance requires visits to accepting states \emph{and} extended Parikh images in $C$.
Thus, e.g., reachability and safety are not dual on a syntactic level. 
Nevertheless, we later prove dualities on a semantic level.

Similarly, one can easily show that a disjunctive definition is equivalent to our conjunctive one: 
One can reflect in the extended Parikh image of a run prefix whether it ends in an accepting state
and then encode acceptance in the semilinear set.
So, any given Parikh automaton~$(\aut,C)$ (with disjunctive or conjunctive acceptance) can be turned into another one~$(\aut', C')$ capturing acceptance in $(\aut, C)$ by Parikh images only. 
So, with empty (full) set of accepting states and $C'$ mimicking disjunction (conjunction), it is equivalent to the original automaton with disjunctive (conjunctive) acceptance.

Now, the language~$\Lall(\aut, C)$ of a safety Parikh automaton (\spa)~$(\aut, C)$ contains those words~$w\in \Sigma^\omega$ such that $(\aut, C)$ has a safety accepting run processing $w$.
Similarly, we define the languages
\begin{itemize}
    \item $\Lsexists(\aut, C)$ of synchronous reachability Parikh automata (\srpa),
    \item $\Laexists(\aut, C)$ of asynchronous reachability Parikh automata (\arpa),
    \item $\LsBuchi(\aut, C)$ of synchronous Büchi Parikh automata (\abpa),
    \item $\LaBuchi(\aut, C)$ of asynchronous Büchi Parikh automata (\sbpa), and
    \item $\LcoBuchi(\aut, C)$ of co-Büchi Parikh automata (\cpa).
\end{itemize}
Determinism for all types of automata is defined as for Parikh automata on finite words.
Unless explicitly stated otherwise, every automaton is assumed to be nondeterministic.

\begin{example}
\label{example_omega}
Let $\aut$ be the \dfa shown in Figure~\ref{figure_balanced} and let $C = \set{(n,n) \mid n \in\nats}$ and $\overline{C} = \set{(n,n') \mid n\neq n'} = \nats^2 \setminus C$.

Recall that a finite word~$w$ is $(a,b)$-balanced if $\occ{w}{a} = \occ{w}{b}$, i.e., the number of $a$'s and $b$'s in $w$ is equal. 
The empty word is $(a,b)$-balanced and every odd-length word over $\set{a,b}$ is $(a,b)$-unbalanced.

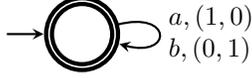
\begin{figure}
    \centering
    \begin{tikzpicture}[ultra thick]
    
    \node[state,accepting] (1) at (0,0) {};

    \path[-stealth]
    (-1,0) edge (1)
    (1) edge[loop right] node[right, align=left] {$a,(1,0)$\\$b,(0,1)$} ()
    ;
    
    \end{tikzpicture}
    \caption{The automaton for Example~\ref{example_omega}.}
    \label{figure_balanced}
\end{figure}
\begin{enumerate}

    \item When interpreting $(\aut, C)$ as a \pa, it accepts the language of finite $(a,b)$-balanced words; when interpreting $(\aut, \overline{C})$ as a \pa, it accepts the language of finite $(a,b)$-unbalanced words.
    
    \item When interpreting $(\aut, C)$ as an \arpa or \srpa, it accepts the language of infinite words that have an $(a,b)$-balanced prefix; when interpreting $(\aut, \overline{C})$ as an \arpa or \srpa, it accepts the language of infinite words that have an $(a,b)$-unbalanced prefix. Note that both languages are universal, as the empty prefix is always $(a,b)$-balanced and every odd-length prefix is $(a,b)$-unbalanced.
    
    \item When interpreting $(\aut, C)$ as an \spa, it accepts the language of infinite words that have only $(a,b)$-balanced prefixes; when interpreting $(\aut, \overline{C})$ as an \spa, it accepts the language of infinite words that have only $(a,b)$-unbalanced prefixes. Here, both languages are empty, which follows from the same arguments as for universality in the previous case.
    
    \item When interpreting $(\aut, C)$ as an \abpa or \sbpa, it accepts the language of infinite words with infinitely many $(a,b)$-balanced prefixes; when interpreting $(\aut, \overline{C})$ as an \abpa or \sbpa, it accepts the language of infinite words with infinitely many $(a,b)$-unbalanced prefixes.
    The latter language is universal, as every odd-length prefix is unbalanced.
    
    \item When interpreting $(\aut, C)$ as a \cpa, it accepts the language of infinite words such that almost all prefixes are $(a,b)$-balanced; when interpreting $(\aut, \overline{C})$ as a \cpa, it accepts the language of infinite words such that almost all prefixes are $(a,b)$-unbalanced. 
    Again, the former language is empty.
\end{enumerate}
\end{example}

Let $(\aut, C)$ be a Parikh automaton. We say that a run prefix is an $F$-prefix if it ends in an accepting state of $\aut$, a $C$-prefix if its extended Parikh image is in $C$, and an $FC$-prefix if it is both an $F$-prefix and a $C$-prefix.
Note that both asynchronous acceptance conditions are defined in terms of the existence of $F$-prefixes and $C$-prefixes and all other acceptance conditions in terms of the existence of $FC$-prefixes.

\begin{remark}
\srpa and \arpa (\spa, \sbpa and \abpa, \cpa) are strictly more expressive than $\omega$-regular reachability  (safety, Büchi, co-Büchi) automata. 
Inclusion follows by definition while strictness is witnessed by the  languages presented in Example~\ref{example_omega}.
\end{remark}

\section{Expressiveness}
\label{sec:expressiveness}

In this section, we study the expressiveness of the various types of Parikh automata on infinite words introduced above, by comparing synchronous and asynchronous variants, deterministic and nondeterministic variants, and the different acceptance conditions.

\begin{remark}
\label{remark:completeness}
In this, and only this, section, we consider only reachability Parikh automata that are complete in the following sense: For every state~$q$ and every letter~$a$ there is a vector~$\vec{v}$ and a state~$q'$ such that $(q,(a,\vec{v}),q')$ is a transition of $\aut$, i.e., every letter can be processed from every state. 
Without this requirement, one can express safety conditions by incompleteness, while we want to study the expressiveness of \myquot{pure} reachability automata.

Safety, Büchi, and co-Büchi automata can be assumed, without loss of generality, to be complete, as one can always add a nonaccepting sink to complete such an automaton without modifying the accepted language.
\end{remark}

We begin our study by comparing the synchronous and asynchronous variants of reachability and Büchi automata.
All transformations proving the following inclusions are effective and lead to a linear increase in the number of states and a constant increase in the dimension of the semilinear sets.

\begin{theorem}\hfill
\label{theorem_synchvsasynch}
\begin{enumerate}
    \item\label{theorem_synchvsasynch_reach} \arpa and \srpa are equally expressive.
    \item\label{theorem_synchvsasynch_detreach} Deterministic \arpa are strictly more expressive than  deterministic \srpa.
    \item\label{theorem_synchvsasynch_buchi} \abpa and \sbpa are equally expressive.
    \item\label{theorem_synchvsasynch_detbuchi} Deterministic \abpa are strictly more expressive than deterministic \sbpa.
\end{enumerate}
\end{theorem}

Due to the equivalence of synchronous and asynchronous (nondeterministic) reachability Parikh automata, we drop the qualifiers whenever possible and just speak of reachability Parikh automata (\rpa). 
We do the same for (nondeterministic) Büchi Parikh automata~(\bpa).

Next, we compare the deterministic and nondeterministic variants for each acceptance condition. 
Note that all separations are as strong as possible, i.e., for reachability and Büchi we consider deterministic asynchronous automata, which are more expressive than their synchronous counterparts (see Theorem~\ref{theorem_synchvsasynch}).

\begin{theorem}\hfill
\label{theorem_detvsnondet}
\begin{enumerate}
    \item \label{theorem_detvsnondet_reach} Nondeterministic \rpa are strictly more expressive than deterministic \arpa.
    
    \item \label{theorem_detvsnondet_safety} Nondeterministic \spa are strictly more expressive than deterministic \spa.
    
    \item \label{theorem_detvsnondet_buchi} Nondeterministic \bpa are strictly more expressive than deterministic \abpa.
    
    \item \label{theorem_detvsnondet_cobuchi} Nondeterministic \cpa are strictly more expressive than deterministic \cpa.
\end{enumerate}
\end{theorem}

After having separated deterministic and nondeterministic automata for all acceptance conditions, we now consider inclusions and separations between the different acceptance conditions.
Here, the picture is notably different than in the classical $\omega$-regular setting, as almost all classes can be separated.

\begin{theorem}
\label{thm_omegaseparations}
Every \rpa can be turned into an equivalent \bpa and into an equivalent \cpa. All other automata types are pairwise incomparable.
\end{theorem}

Our separations between the different acceptance conditions are as strong as possible, e.g., when we show that not every \rpa has an equivalent \spa, we exhibit a \emph{deterministic} \srpa (the weakest class of \rpa) whose language is not accepted by any nondeterministic \spa (the strongest class of \spa).
The same is true for all other separations.

\section{Closure Properties}
\label{sec:closure}

In this section, we study the closure properties of Parikh automata on infinite words. 
We begin by showing that, for deterministic synchronous automata, reachability and safety acceptance as well as Büchi and co-Büchi acceptance are dual, although they are not syntactically dual due to all acceptance conditions being defined by a conjunction.
On the other hand, deterministic asynchronous automata can still be complemented, however only into nondeterministic automata.

\begin{theorem}
\label{theorem_detcomplementation}
\hfill
\begin{enumerate}
    
    \item \label{theorem_detcomplementation_sreach2safety} Let $(\aut, C)$ be a deterministic~\srpa. The complement of $\Lsexists(\aut, C)$ is accepted by a deterministic~\spa.

    \item \label{theorem_detcomplementation_areach2safety} Let $(\aut, C)$ be a deterministic~\arpa. The complement of $\Laexists(\aut, C)$ is accepted by an \spa, but not necessarily by a deterministic \spa.
    
    \item \label{theorem_detcomplementation_safety2reach} Let $(\aut, C)$ be a deterministic~\spa. The complement of $\Lall(\aut, C)$ is accepted by a deterministic~\srpa.
    
    \item \label{theorem_detcomplementation_sbuchi2cobuchi} Let $(\aut, C)$ be a deterministic~\sbpa. The complement of $\LsBuchi(\aut, C)$ is accepted by a deterministic~\cpa.

    \item \label{theorem_detcomplementation_abuchi2cobuchi} Let $(\aut, C)$ be a deterministic~\abpa. The complement of $\LaBuchi(\aut, C)$ is accepted by a \cpa, but not necessarily by a deterministic \cpa.

    \item \label{theorem_detcomplementation_cobuchi2buchi} Let $(\aut, C)$ be a deterministic~\cpa. The complement of $\LcoBuchi(\aut, C)$ is accepted by a deterministic~\sbpa.

\end{enumerate}
\end{theorem}

The positive results above are for deterministic automata. 
For nondeterministic automata, the analogous statements fail.

\begin{theorem}
\label{theorem_nondetcomplementation}
\hfill
\begin{enumerate}
    
    \item \label{theorem_nondetcomplementation_reach2safety} 
    There exists an \srpa~$(\aut, C)$ such that no \spa accepts
    the complement of ${\Lsexists(\aut, C)}$.
    
    \item \label{theorem_nondetcomplementation_safety2reach} 
    There exists an \spa~$(\aut, C)$ such that no \rpa accepts the complement of ${\Lall(\aut, C)}$.
    
    \item \label{theorem_nondetcomplementation_buchi2cobuchi} 
    There exists an \sbpa~$(\aut, C)$ such that no \cpa accepts the complement of ${\LsBuchi(\aut, C)}$.
    
    \item \label{theorem_nondetcomplementation_cobuchi2buchi} 
    There exists a \cpa~$(\aut, C)$ such that no \bpa accepts the complement of ${\LcoBuchi(\aut, C)}$.
    
    \end{enumerate}
\end{theorem}

    
    
    
    
    

Next, we consider closure under union, intersection, and complementation of the various classes of Parikh automata on infinite words.
Notably, all nondeterministic (and some deterministic) classes are closed under union, the picture for intersection is more scattered, and we prove failure of complement closure for all classes. 
Again, this is in sharp contrast to the setting of classical Büchi automata, which are closed under all three Boolean operations.

\begin{table}
 
    \centering
    \begin{tabular}{lcccccccc}
\toprule
& \multicolumn{3}{c}{Closure} && \multicolumn{4}{c}{Decision Problems}\\
& $\cup$ & $\cap$ & $\overline{\phantom{x}}$ &\qquad{} & Nonemptiness & Universality & Model Check.\  & Games\\
\midrule
\rpa          & \yess & \yess & \noo && \NP-compl.\ & undec.\ & undec.\  &undec.\ \\
det.\ \arpa   & \noo & \noo & \noo && \NP-compl.\ & undec.\ & undec.\ &undec.\ \\
det.\ \srpa   & \yess & \noo & \noo && \NP-compl.\ & undec.\ & undec.\ &undec.\ \\
\midrule
\spa          & \yess & \yess & \noo && undec.\ & undec.\ & undec.\ & undec.\ \\
det.\ \spa    & \noo & \yess & \noo && undec.\ & \coNP-compl.\ & \coNP-compl.\ & undec.\ \\
\midrule
\bpa          & \yess & \noo & \noo && \NP-compl.\ & undec.\ & undec.\ &undec.\ \\
det.\ \abpa   & \quest & \noo & \noo && \NP-compl.\ & undec.\ & undec.\ &undec.\ \\
det.\ \sbpa   & \yess & \noo & \noo && \NP-compl.\ & undec.\ & undec.\ &undec.\ \\
\midrule
\cpa          & \yess & \yess & \noo && undec.\ & undec.\ & undec.\ & undec.\ \\
det.\ \cpa    & \noo & \yess & \noo && undec.\ & \coNP-compl.\ & \coNP-compl.\ & undec.\ \\
\toprule
\end{tabular}
    \caption{Closure properties and decidability of decision problems for Parikh automata on infinite words.}
    \label{table_results}
\end{table}

\begin{theorem}
\label{thm_omegaclosureprops}
The closure properties depicted in Table~\ref{table_results} hold.
\end{theorem}

Note that there is one question mark in the closure properties columns in Table~\ref{table_results}, which we leave for further research.

\section{Decision Problems}
\label{sec:decision}
In this section, we study the complexity of the nonemptiness and the universality problem, model checking, and solving games for Parikh automata on infinite words.
Before we can do so, we need to specify how a Parikh automaton~$(\aut, C)$ is represented as input for algorithms: The vectors labeling the transitions of $\aut$ are represented in binary and a linear set
\[
\left\{ \vec{v}_0 + \sum\nolimits_{i=1}^k c_i\vec{v}_i \:\middle|\: c_i \in \nats \text{ for } i=1,\ldots, k \right\}
\]
is represented by the list~$(\vec{v}_0, \ldots, \vec{v}_k)$ of vectors, again encoded in binary.
A semilinear set is then represented by a set of such lists.

\subsection{Nonemptiness}
\label{subsec:decision1}


We begin by settling the complexity of the nonemptiness problem.
The positive results are obtained by reductions to the nonemptiness of Parikh automata on finite words
while the undecidability results are reductions from
the termination problem for two-counter machines.


\begin{theorem}\label{thm_reachsafetyemptiness} 
The following problems are \NP-complete: 
\begin{enumerate}
    \item\label{rpaemptiness} Given an \rpa, is its language nonempty?
\item\label{bpaemptiness} Given a \bpa, is its language nonempty?
\end{enumerate}The following problems are undecidable:
\begin{enumerate}
\setcounter{enumi}{2}
        \item\label{spaemptiness} Given a deterministic \spa, is  its language nonempty?
    \item\label{cpaemptiness} Given a deterministic \cpa, is its language nonempty?
\end{enumerate}
\end{theorem}

\begin{proof}
\ref{rpaemptiness}.)
Due to Theorem~\ref{theorem_synchvsasynch}.\ref{theorem_synchvsasynch_reach}, we only consider the case of \srpa for the \NP upper bound. 
Given such an automaton~$(\aut, C)$ with $\aut = (Q, \Sigma\times D, q_\init, \Delta, F)$ let $F' \subseteq F$ be the set of accepting states from which a cycle is reachable.
Now, define $\aut'= (Q, \Sigma\times D, q_\init, \Delta, F')$.
Then, we have $\Lsexists(\aut, C) \neq \emptyset$ if and only if $L(\aut', C) \neq \emptyset$ (i.e., we treat $(\aut', C)$ as a \pa), with the latter problem being in \NP~\cite{FL}.

The matching \NP lower bound is, again due to Theorem~\ref{theorem_synchvsasynch}.\ref{theorem_synchvsasynch_reach}, only shown for \srpa. 
We proceed by a reduction from the \NP-complete~\cite{FL} nonemptiness problem for Parikh automata.
Given a Parikh automaton~$(\aut, C)$, let $\aut'$ be obtained from $\aut$ by adding a fresh state~$q$ with a self-loop labeled by $(\#, \vec{0})$ as well as transitions labeled by $(\#, \vec{0})$ leading from the accepting states of $\aut$ to $q$.
Here, $\#$ is a fresh letter and $\vec{0}$ is the zero vector of the correct dimension.
By declaring $q$ to be the only accepting state in $\aut'$, we have that $\Lsexists(\aut', C)$ is nonempty if and only if $L(\aut, C)$ is nonempty.

Note that hardness holds already for deterministic automata, as one can always rename letters to make a nondeterministic \pa deterministic without changing the answer to the nonemptiness problem. 

\ref{bpaemptiness}.)
Due to Theorem~\ref{theorem_synchvsasynch}.\ref{theorem_synchvsasynch_buchi}, it is enough to consider synchronous Büchi acceptance for the upper bound. 
So, fix some \sbpa~$(\aut, C)$ with $\aut = (Q, \Sigma\times D, q_\init, \Delta, F)$.
Let $C = \bigcup_{i}L_i$ where the $L_i$ are linear sets. 
The language $\LsBuchi(\aut, C)$ is nonempty if and only if $\LsBuchi(\aut, L_i)$ is nonempty for some $i$. 
Hence, we show how to solve nonemptiness for automata with linear~$C$, say $C = \left\{ \vec{v}_0 + \sum\nolimits_{i=1}^k c_i\vec{v}_i \:\middle|\: c_i \in \nats \text{ for } i=1,\ldots, k \right\}.
$
We define 
$
P = \left\{ \sum\nolimits_{i=1}^k c_i\vec{v}_i \:\middle|\: c_i \in \nats \text{ for } i=1,\ldots, k \right\}
$
and, for a given state~$q \in Q$, the \nfa
\begin{itemize}
    \item $\aut_q$ obtained from $\aut$ by replacing the set of accepting states by $\set{q}$, and
    \item $\aut_{q,q}$ obtained from $\aut$ by replacing the initial state by $q$, by replacing the set of accepting states by $\set{q}$, and by modifying the resulting \nfa such that it does not accept the empty word (but leaving its language unchanged otherwise).
\end{itemize}
We claim that $\LsBuchi(\aut, C)$ is nonempty if and only if there is a $q \in F$ such that both $L(\aut_q, C)$ and $L(\aut_{q,q},P)$ are nonempty.
As nonemptiness of Parikh automata is in \NP, this yields the desired upper bound.

So, assume there is such a $q$.
Then, there is a finite run~$\rho_1$ of $\aut$ that starts in $q_\init$, ends in $q$, and processes some~$w_1 \in \Sigma^*$ with extended Parikh image in $C$.
Also, there is a finite run~$\rho_2$ of $\aut$ that starts and ends in $q$ and processes some nonempty $w_2 \in \Sigma^*$ with extended Parikh image in $P$.
For every $n\geqslant 1$, $\rho_1(\rho_2)^n$ is a finite run of $(\aut,C)$ ending in the accepting state~$q$ that processes $w_1(w_2)^n$ and whose extended Parikh image is in $C$. 
So, $\rho_1(\rho_2)^\omega$ is a synchronous Büchi accepting run of $(\aut, C)$. 

For the converse direction, assume that there is some synchronous Büchi accepting run~$(q_0, w_0, q_1)(q_1, w_1, q_2)(q_2, w_2, q_3) \cdots$ of $(\aut, C)$. 
Then, there is also an accepting state~$q \in F$ and an infinite set of positions~$S \subseteq \nats$ such that $q_s = q$ and $\parikhimage(w_0 \cdots w_{s-1}) \in C$ for all $s \in S$.
Hence, for every $s \in S$ there is a vector~$(c_1^s, \ldots, c_k^s) \in \nats^k$ such that
$\parikhimage(w_0 \cdots w_{s-1}) = \vec{v}_0 + \sum\nolimits_{i=1}^k c_i^s\vec{v}_i$.
By Dickson's Lemma~\cite{Dickson13}, there are $s_1 < s_2$ such that $c_j^{s_1} \le c_j^{s_2}$ for every $1 \le j \le k$.
Then, $
\parikhimage(w_{s_1} \cdots w_{s_2-1})  = \sum\nolimits_{i=1}^k (c_i^{s_2} - c_i^{s_1})\vec{v}_i$,
which implies $\parikhimage(w_{s_1} \cdots w_{s_2-1}) \in P$.
Thus, the prefix of $\rho$ of length~$s_1$ is an accepting run of the \pa~$(\aut_q, C)$ and the next~$(s_2 - s_1)$ transitions of $\rho$ form a nonempty accepting run of the \pa~$(\aut_{q,q}, P)$.

The \NP lower bound follows from the proof of Theorem~\ref{thm_reachsafetyemptiness}.\ref{rpaemptiness} by noticing that we also have that $\LsBuchi(\aut',C)$ is nonempty if and only if $L(\aut, C)$ is nonempty.

\ref{spaemptiness}.) and  
\ref{cpaemptiness}.)
The two undecidability proofs,
based on reductions from
undecidable problems for two-counter machines,
are relegated to the appendix, which introduces all the required technical details on such machines.
\end{proof}

We conclude Subsection~\ref{subsec:decision1}
by displaying an interesting consequence of 
the decomposition used in
the proof of Theorem~\ref{thm_reachsafetyemptiness}.\ref{bpaemptiness}:
we show that the language of every \bpa (on infinite words)
can be expressed by combining the languages
of well-chosen \pa (on finite words).
This is similar to what happens in other settings:
For instance, $\omega$-regular languages are exactly the languages of the form~$\bigcup_{j=1}^n L_j \cdot (L_j')^\omega$, where each $L_j, L_j'$ is regular.
Analogously, $\omega$-context-free languages can be characterized by context-free languages~\cite{DBLP:journals/jcss/CohenG77}.
For Büchi Parikh automata, one direction of the above characterization holds:

\begin{lemma}
\label{lemma_buchilike}
If a language~$L$ is accepted by a \bpa then $L = \bigcup_{j=1}^n L_j \cdot (L_j')^\omega$, where each $L_j, L_j'$ is accepted by some \pa.
\end{lemma}

\begin{proof}
Let $L$ be accepted by an \sbpa~$(\aut, C)$ with $C = \bigcup_{j \in J} C_j$ for some finite set~$J$, where each $C_j$ is a linear set.
In the proof of Theorem~\ref{thm_reachsafetyemptiness}.\ref{bpaemptiness}, we have defined the \nfa~$\aut_q$ and $\aut_{q,q}$ for every state~$q$ of $\aut$.

Now, consider some $w \in L$ and an accepting run~$\rho = (q_0, w_0, q_1)(q_1, w_1, q_2)(q_2, w_2, q_3)\cdots $ of $(\aut, C)$ processing $w$.
Then, there is an accepting state~$q$ of $\aut$, some $C_j$, and an infinite set~$S \subseteq \nats$ of positions such that $q_s = q$ and $\parikhimage(w_0\cdots w_{s-1}) \in C_j$ for all $s \in S$.
Let $C_j = \left\{ \vec{v}_0 + \sum\nolimits_{i=1}^k c_i\vec{v}_i \:\middle|\: c_i \in \nats \text{ for } i=1,\ldots, k \right\}$ and let $P_j = \left\{\sum\nolimits_{i=1}^k c_i\vec{v}_i \:\middle|\: c_i \in \nats \text{ for } i=1,\ldots, k \right\}$.
As before, for every $s \in S$, there is a vector~$(c_1^s,\ldots, c_k^s) \in \nats^k$ such that $\parikhimage(w_0\cdots w_{s-1}) = \vec{v}_0 + \sum\nolimits_{i=1}^k c_i^s\vec{v}_i$.

Now, we apply an equivalent formulation of Dickson's Lemma~\cite{Dickson13}, which yields an infinite subset~$S' \subseteq S$ with $ c_j^s \le c_j^{s'} $ for all $1 \le j\le k$ and all $s, s' \in S'$ with $s < s'$, i.e., we have an increasing chain in $S$.
Let $s_0 < s_1 < s_2 < \cdots$ be an enumeration of $S'$.

As above, $\parikhimage(w_{s_n} \cdots w_{s_{n+1}-1}) \in P_j$ for all $n$.
So, $(q_0, w_0, q_1) \cdots (q_{s_0-1}, w_{s_0-1}, q_{s_0})$ is an accepting run of the \pa~$(\aut_q,C)$ and each $(q_{s_n}, w_{s_n}, q_{s_n+1}) \cdots (q_{s_{n+1}-1}, w_{s_{n+1}-1}, q_{s_{n+1}})$ is an accepting run of the \pa~$(\aut_{q,q},P)$.
So, $w \in \bigcup_{j\in J} \bigcup_{q \in Q}L(\aut_q, C_j) \cdot (L(\aut_{q,q},P_j))^\omega$.
\end{proof}

Recall that a word is ultimately periodic if it is of the form~$xy^\omega$.
Every nonempty $\omega$-regular and every nonempty $\omega$-context-free language contains an ultimately periodic word, which is a simple consequence of them being of the form~$\bigcup_{j=1}^n L_j \cdot (L_j')^\omega$.

\begin{corollary}
Every nonempty language accepted by a \bpa contains an ultimately periodic word.
\end{corollary}

Let us briefly comment on the other direction of the implication stated in Lemma~\ref{lemma_buchilike}, i.e., is every language of the form~$ \bigcup_{j=1}^n L_j \cdot (L_j')^\omega$, where each $L_j, L_j'$ is accepted by some \pa, also accepted by some \bpa?
The answer is no:
Consider $L = \set{a^nb^n \mid n > 1}$, which is accepted by a deterministic \pa.
However, using the shifting technique (see Remark~\ref{remark:shifting}), one can show that $L^\omega$ is not accepted by any \bpa: 
Every accepting run of an $n$-state \bpa processing~$(a^nb^n)^\omega$ can be turned into an accepting run on a word of the form~$(a^nb^n)^* a^{n+k} b^n (a^nb^n)^* a^{n-k} b^n (a^n b^n)^\omega $ for some $k > 0$ by shifting some cycle to the front while preserving Büchi acceptance.

For reachability acceptance, a similar characterization holds, as every \rpa can be turned into an equivalent \bpa.
But for safety and co-Büchi acceptance the characterization question is nontrivial, as for these acceptance conditions all (almost all) run prefixes have to be $FC$-prefixes.
We leave this problem for future work.

\subsection{Universality}
\label{subsec:decision2}

Now, we consider the universality problem.
Here, the positive results follow from the duality of deterministic \spa and \rpa (\cpa and \bpa) and the decidability of nonemptiness for the dual automata classes.
Similarly, the undecidability proofs for deterministic \srpa and \sbpa follow from duality and undecidability of nonemptiness for the dual automata classes. 
Finally, the remaining undecidability results follow from reductions from undecidable problems for two-counter machines and Parikh automata over finite words.

\begin{theorem}
\label{theorem_universality} The following problems are \coNP-complete:
\begin{enumerate}
    \item \label{theorem_universality_detsafety}  Given a deterministic \spa, is its language universal?
    
    \item \label{theorem_universality_detcobuchi} Given a deterministic \cpa, is its language universal?
\end{enumerate}
The following problems are undecidable:
\begin{enumerate}\setcounter{enumi}{2}
    \item \label{theorem_universality_reach}  Given a deterministic \srpa, is its language universal?
    
    \item \label{theorem_universality_nondetsafety}Given an \spa, is its language universal?
    
    \item \label{theorem_universality_buchi}Given a deterministic \sbpa, is its language universal?
    
    \item \label{theorem_universality_nondetcobuchi} Given a  \cpa, is its language universal?
\end{enumerate}

\end{theorem}

\begin{proof}
The proofs of the results for deterministic automata follow immediately from the fact that a language is universal if and only if its complement is empty, Theorem~\ref{theorem_detcomplementation}, and Theorem~\ref{thm_reachsafetyemptiness}.
The proof of undecidability for \spa,
based on a reduction from
the termination problem for two-counter machines,
is relegated to the appendix
where all the necessary technical details are presented.
To conclude, let us consider universality of \cpa.

\ref{theorem_universality_nondetcobuchi}.)
Universality of Parikh automata over finite words is undecidable~\cite{KR}.
Now, given a \pa~$(\aut, C)$ over $\Sigma$, one can construct a \cpa~$(\aut', C')$ for the language~$L(\aut,C)\cdot\#\cdot(\Sigma\cup\set{\#})^\omega \cup \Sigma^\omega$, where $\#\notin \Sigma$ is a fresh letter.
This construction relies on freezing the counters (i.e., moving to a copy of the automaton with the same transition structure, but where the counters are no longer updated) and closure of \cpa under union.
Now, $L(\aut, C)$ is universal if and only if $\LcoBuchi(\aut', C')$ is universal.
\end{proof}

\subsection{Model Checking}

Model checking is arguably the most successful application of automata theory to automated verification.
The problem asks whether a given system satisfies a specification, often given by an automaton. 

More formally, and for the sake of notational convenience, we say that a transition system~$\trans$ is a (possibly incomplete) \spa $(\aut, C)$ so that every state of $\aut$ is accepting and $C = \nats^d$, i.e., every run is accepting.
Now, the model-checking problem for a class $\Lcal$ of languages of infinite words asks, given a transition system~$\trans$ and a language~$L \in \Lcal$, whether $\Lall(\trans) \subseteq L$, i.e., whether every word in the transition system satisfies the specification~$L$.
Note that our definition here is equivalent to the standard definition of model checking of finite-state transition systems. 

Here, we study the model-checking problem for different types of Parikh automata.

\begin{theorem}
\label{theorem_modelchecking}
The following problems are \coNP-complete:
\begin{enumerate}

\item \label{theorem_modelchecking_detsafety}
     Given a transition system~$\trans$ and a deterministic \spa~$(\aut, C)$, is $\Lall(\trans) \subseteq \Lall(\aut, C)$?
    
\item \label{theorem_modelchecking_detcobuchi}
    Given a transition system~$\trans$ and a deterministic \cpa~$(\aut, C)$, is $\Lall(\trans) \subseteq \LcoBuchi(\aut, C)$?
    
\end{enumerate}
The following problems are undecidable:
\begin{enumerate}
\setcounter{enumi}{2}
    \item \label{theorem_modelchecking_reach}
     Given a transition system~$\trans$ and a deterministic \srpa~$(\aut, C)$, is $\Lall(\trans) \subseteq \Lsexists(\aut, C)$?

\item \label{theorem_modelchecking_safety}
    Given a transition system~$\trans$ and an \spa~$(\aut, C)$, is $\Lall(\trans) \subseteq \Lall(\aut, C)$?
    
\item \label{theorem_modelchecking_buchi}
    Given a transition system~$\trans$ and a deterministic \sbpa~$(\aut, C)$, is $\Lall(\trans) \subseteq \LsBuchi(\aut, C)$?
    
\item \label{theorem_modelchecking_nondetcobuchi}
    Given a transition system~$\trans$ and a \cpa~$(\aut, C)$, is $\Lall(\trans) \subseteq \LcoBuchi(\aut, C)$?
\end{enumerate}
\end{theorem}

\begin{proof}
Let $\trans$ be a transition system with $\Lall(\trans) = \Sigma^\omega$, e.g., a one-state transition system with a self-loop labeled with all letters in $\Sigma$. 
Then, $L \subseteq\Sigma^\omega$ is universal if and only if $L$ $\Lall(\trans) \subseteq L$. 
Thus, all six lower bounds (\coNP-hardness and undecidability) immediately follow from the analogous lower bounds for universality (see Theorem~\ref{theorem_universality}).
So, it remains to consider the two \coNP upper bounds.

So, fix a deterministic~\spa~$(\aut, C)$ and a transition system~$\trans$. 
We apply the usual approach to automata-theoretic model checking: We have $\Lall(\trans)\subseteq \Lall(\aut, C)$ if and only if $\Lall(\trans) \cap \overline{\Lall(\aut, C)} = \emptyset$.
Due to Theorem~\ref{theorem_detcomplementation}.\ref{theorem_detcomplementation_safety2reach} there is a deterministic~\srpa~$(\aut', C')$ accepting $\overline{\Lall(\aut, C)}$. 
Furthermore, using a product construction, one can construct an \rpa~$(\aut'', C'')$ accepting $\Lall(\trans) \cap \overline{\Lall(\aut, C)}$, which can then be tested for emptiness, which is in \coNP (see Theorem~\ref{thm_reachsafetyemptiness}).
Note that the product construction depends on the fact that every run of $\trans$ is accepting, i.e., the acceptance condition of the product automaton only has to check one acceptance condition.

The proof for deterministic co-Büchi Parikh automata is analogous, but using Büchi automata~$(\aut', C')$ and $(\aut'', C'')$.
\end{proof}
\subsection{Infinite Games}
\label{subsec_games}

In this section, we study infinite games with winning conditions specified by Parikh automata. 
Such games are the technical core of the synthesis problem, the problem of determining whether there is a reactive system satisfying a given specification on its input-output behavior. 
Our main result is that solving infinite games is undecidable for all acceptance conditions we consider here.

Here, we consider Gale-Stewart games~\cite{GaleStewart53}, abstract games induced by a language~$L$ of infinite words, in which two players alternately pick letters, thereby constructing an infinite word~$w$. 
One player aims to ensure that $w$ is in $L$ while the other aims to ensure that it is not in $L$.
Formally, given a language~$L\subseteq (\Sigma_1\times \Sigma_2)^\omega$, the game $G(L)$ is played between Player~1 and Player~2 in rounds $i=0,1,2,\ldots$ as follows: At each round~$i$, first Player~1 plays a letter $a_i\in\Sigma_1$ and then Player~2 answers with a letter $b_i\in \Sigma_2$. 
A play of $G(L)$ is an infinite outcome~$w = \binom{a_0}{b_0}\binom{a_1}{b_1}\cdots$ and Player~2 wins it if and only if $w \in L$.

A strategy for Player~$2$ in $G(L)$ is a mapping from $\Sigma_1^+$ to $\Sigma_2$ that gives for each prefix played by Player~1 the next letter to play. 
An outcome~$\binom{a_0}{b_0}\binom{a_1}{b_1}\cdots$ agrees with a strategy $\sigma$ if for each $i$, we have that $b_i=\sigma(a_0a_1\dots a_i)$. 
Player~2 wins $G(L)$ if she has a strategy that only agrees with outcomes that are winning for Player~2.

The next result follows immediately from the fact that for all classes of deterministic Parikh automata, either nonemptiness or universality is undecidable, and that these two problems can be reduced to solving Gale-Stewart games.

\begin{theorem}
\label{thm_dpagamesundec}
The problem \myquot{Given an  automaton~$(\aut, C)$, does Player~2 win $G(L(\aut, C))$?} is undecidable for the following classes of automata: (deterministic) \srpa, (deterministic) \spa, (deterministic) \sbpa, and (deterministic) \cpa.
\end{theorem}

\begin{proof}
The results follow immediately from the following two facts and the undecidability of nonemptiness or universality for the corresponding automata types. Fix a language~$L$.

\begin{itemize}
    \item Player~2 wins $G(\binom{\#}{L})$ with
$\binom{\#}{L} = \set{ \binom{\#}{w_0}  \binom{\#}{w_{1}}\binom{\#}{w_{2}} \cdots  \mid   w_0  w_{1}w_2 \cdots \in L }$ if and only if $L$ is nonempty.
    
    \item Player~2 wins $G(\binom{L}{\#})$ with
$\binom{L}{\#} = \set{ \binom{w_0}{\#}  \binom{w_{1}}{\#}\binom{w_{2}}{\#} \cdots  \mid   w_0  w_{1}w_2 \cdots \in L}$ if and only if $L$ is universal.
\end{itemize}
To conclude, note that a Parikh automaton for $L$ can be turned into an equivalent one for $\binom{\#}{L}$ and $\binom{L}{\#}$ while preserving determinism and the acceptance type, by just replacing each transition label~$a$ by $\binom{\#}{a}$ and $\binom{a}{\#}$, respectively.
\end{proof}

\section{Conclusion}
\label{sec_conc}

In this work, we have extended Parikh automata to infinite words and studied expressiveness, closure properties, and decision problems. 
Unlike their $\omega$-regular counterparts, Parikh automata on infinite words do not form a nice hierarchy induced by their acceptance conditions. 
This is ultimately due to the fact that transitions cannot be disabled by the counters running passively along a run. 
Therefore, a safety condition on the counters cannot be turned into a, say, Büchi condition on the counters, something that is trivial for state conditions.
Furthermore, we have shown that nonemptiness, universality, and model checking are decidable for some of the models we introduced, but undecidable for others. 
Most importantly, we prove \coNP-completeness of model checking with specifications given by deterministic Parikh automata with safety and co-Büchi acceptance.
This allows for the automated verification of quantitative safety and persistence properties.
Finally, solving infinite games is undecidable for all models.

Note that we have \myquot{only} introduced reachability, safety, Büchi, and co-Büchi Parikh automata. 
There are  many more acceptance conditions in the $\omega$-regular setting, e.g., parity, Rabin, Streett, and Muller. 
We have refrained from generalizing these, as any natural definition of these acceptance conditions will subsume co-Büchi acceptance, and therefore have an undecidable nonemptiness problem.

In future work, we aim to close the open closure property in Table~\ref{table_results}.
Also, we leave open the complexity of the decision problems in case the semilinear sets are not given by their generators, but by a Presburger formula.

One of the appeals of Parikh automata over finite words is their robustness: they can equivalently be defined via a quantitative variant of WMSO, via weighted automata, and other models (see the introduction for a more complete picture).
In future work, we aim to provide similar alternative definitions for Parikh automata on infinite words, in particular, comparing our automata to blind multi-counter automata~\cite{DBLP:journals/fuin/FernauS08} and reversal-bounded counter machines~\cite{Ibarra78}.
However, let us mention that the lack of closure properties severely limits the chances for a natural fragment of MSO being equivalent to Parikh automata on infinite~words.

Let us conclude with the following problem for further research: If a Parikh automaton with, say safety acceptance, accepts an $\omega$-regular language, is there then an equivalent $\omega$-regular safety automaton?
Stated differently, does Parikhness allow to accept more $\omega$-regular languages? The same question can obviously be asked for other acceptance conditions as~well.

\paragraph*{Acknowledgements}
We want to thank an anonymous reviewer for proposing Lemma~\ref{lemma_buchilike}.

\bibliography{biblio}

\clearpage
\appendix
\section*{Appendix}

This appendix contains all proofs omitted due to space restrictions.

\section{Proofs omitted in Section~\ref{sec:expressiveness}}

\subsection{Proof of Theorem~\ref{theorem_synchvsasynch}}

\begin{proof}
\ref{theorem_synchvsasynch_reach}.)
First, let us show that every \arpa can be turned into an equivalent \srpa, i.e., we need to synchronize an $F$-prefix and a $C$-prefix.
To this end, we add two Boolean flags~$f_\text{acc}$ and $f_\text{frz}$ to the state space of $\aut$ to obtain $\aut'$.
The flag~$f_\text{acc}$ is raised once an accepting state has been visited while the flag~$f_\text{frz}$ can be nondeterministically raised at any time during a run, with the effect that the extended Parikh image is frozen, i.e., all subsequent transitions are labeled with a zero vector. 
Thus, if $(\aut, C)$ has an asynchronous reachability accepting run, then $(\aut', C)$ will have a (synchronous) reachability-accepting run which is obtained by freezing the extended Parikh image, if the extended Parikh image is in $C$ before an accepting state is visited for the first time.
On the other hand, if an accepting state is visited before the extended Parikh image is in $C$, then we do not have to freeze, as a state in $\aut'$ is accepting as long as the flag~$f_\text{acc}$ is equal to one.

Formally, let $\aut = (Q, \Sigma\times D, q_\init,\Delta, F)$ and let $\ind{F} \colon Q \rightarrow \set{0,1}$ be the indicator function for $F$. We define $\aut' = (Q \times \set{0,1} \times \set{0,1}, \Sigma\times (D \cup \set{\vec{0}}), (q_\init, \ind{F}(q_\init) , 0), \Delta', Q \times \set{1}\times\set{0,1})$ with     \begin{align*}
    \Delta' = {}&{}\set{ ((q,f_{\text{acc}},0),(a,\vec{v}),(q',\max\set{\ind{F}(q'),f_\text{acc}},0)  \mid (q,(a,\vec{v}),q') \in\Delta } \cup\\
    {}&{}\set{ ((q,f_{\text{acc}},f_{\text{frz}}),(a,\vec{0}),(q',\max\set{\ind{F}(q'),f_\text{acc}},1)  \mid (q,(a,\vec{v}),q') \in\Delta \text{ and } f_{\text{frz}} \in\set{0,1}},
\end{align*}
where $\vec{0}$ is the zero vector of appropriate dimension.
Then, we have $\Lsexists(\aut', C) = \Laexists(\aut, C)$. 

Now, consider the other inclusion, i.e., we want to turn an \srpa into an equivalent \arpa.
Here, we reflect whether the last state of a run prefix is accepting or not in the extended Parikh image of the run prefix. 
As we consider synchronous reachability acceptance, the visit of an accepting state and the Parikh image being in the semilinear set happen at the same time, i.e., they can be captured just by a semilinear set using the additional information in the extended Parikh image. 
As now both requirements are captured by the semilinear set, we can just make every state accepting to obtain an equivalent \arpa.

Formally, let $(\aut, C)$ be an \srpa with $\aut = (Q, \Sigma\times D, q_\init,\Delta, F)$.
If $q_\init$ is in $F$ and the zero vector is in $C$, then $\Lsexists(\aut, C) = \Sigma^\omega$ due to completeness and we have $\Lsexists(\aut, C) = \Laexists(\aut, C)$, i.e., the transformation is trivial. 
So, in the following we assume that $q_\init \notin F$ or that the zero vector is not in $C$.
This implies that if a run is accepting, then this is witnessed by a \emph{nonempty} $FC$-prefix.
Hence, if $q_\init$ is in $F$, then we can add a new nonaccepting initial state to $\aut$ without any incoming transitions, with the same outgoing transitions as the original initial state. 
The resulting automaton is equivalent to $(\aut, C)$ and has a nonaccepting initial state.

Now, define $\aut' = (Q, \Sigma \times D', q_\init, \Delta', Q)$ where 
\[D' = \set{(v_0, \ldots, v_{d-1}, b) \mid (v_0, \ldots, v_{d-1}) \in D \text{ and } b\in\set{0,1} }\]
and $\Delta'$ contains the transition
\[
(q, (a, (v_0, \ldots, v_{d-1},b)), q')
\]
for every transition~$(q, (a,(v_0,\ldots, v_{d-1})),q') \in \Delta$. Here, $b = 0$ if both $q$ and $q'$ are accepting or both are nonaccepting, otherwise $b=1$.

Note that there is a bijection between run prefixes in $\aut$ and run prefixes in $\aut'$, as we only added a new component to the vector labelling each transition.
Now, an induction shows that a nonempty run prefix in $\aut'$ has an extended Parikh image whose last (new) component is odd if and only if the run prefix ends in a state in $F$.
Thus, consider the semilinear set
\[C' = \set{(v_0, \ldots, v_d) \mid (v_0, \ldots, v_{d-1}) \in C \wedge v_d \text{ is odd}}.\]
Then, we have $\Laexists(\aut', C) = \Lsexists(\aut, C)$, as every state of $\aut'$ is accepting.

\ref{theorem_synchvsasynch_detreach}.)
The transformation of an \srpa into an equivalent \arpa presented in Item~\ref{theorem_synchvsasynch_reach} preserves determinism, so every deterministic~\srpa can be turned into a deterministic~\arpa.

Now, we show the strictness of the inclusion.
Consider the language~$L$ of infinite words over~$\set{a,b}$ that have a nonempty $(a,b)$-balanced prefix and contain at least one $c$, which is accepted by the deterministic \arpa~$(\aut, C)$ with $\aut$ depicted in Figure~\ref{fig_balanced_and_c} and $C = \set{(n,n) \mid n > 0}$.

\begin{figure}
    \centering
    \begin{tikzpicture}[ultra thick]
    
    \node[state] (1) at (0,0) {};
    \node[state, accepting] (3) at (3,0) {};
    
    \path[-stealth]
    (-1,0) edge (1)
    (1) edge[loop above] node[align = left,above] {$a,(1,0)$\\$b,(0,1)$}()
    (1) edge[] node[above] {$c,(0,0)$} (2)
    (2) edge[loop right] node[align = left,right] {$a,(1,0)$\\$b,(0,1)$\\$c,(0,0)$} ()
    ;
    
    \end{tikzpicture}
    \caption{The automaton for Theorem~\ref{theorem_synchvsasynch}.\ref{theorem_synchvsasynch_detreach}.}
    \label{fig_balanced_and_c}
\end{figure}
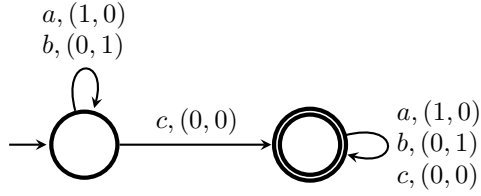

Now, assume there is a deterministic \srpa~$(\aut, C)$ accepting $L$, say with $n$ states.
Consider the word~$w = a(a^{n}b^{n})^{n+1} c^\omega$, which is not in $L$ as it has no $(a,b)$-balanced prefix.
Due to completeness\footnote{Note that $L$ is a liveness property, i.e., every finite word is a prefix of some word in $L$. Hence, any automaton accepting $L$ must be complete. Thus, our separation argument here can even be generalized to incomplete automata. In fact, most, but not all, of our separations the invocation of completeness as an assumption is just for convenience and could be replaced by arguing that the language witnessing the separation is a liveness property.}, $(\aut,C)$ has a run~$\rho$ processing $w$, which is not synchronously reachability accepting. 

The run~$\rho$ processing $w$ contains $n+1$ cycles in the infixes processing the $a^{n}$-infixes. 
Hence, two such cycles start (and thus end) in the same state. 
Let $\rho = \rho_0 \rho_1 \rho_2 \rho_3 \rho'$ be the decomposition of $\rho$ such that $\rho_1$ and $\rho_3$ are these two cycles. 
As the cycles start and end in the same state, $\rho_s = \rho_0 \rho_2 \rho_1 \rho_3 \rho'$ is also a run of $(\aut,C)$, i.e., we shift the first cycle further back.

The first cycle processes a word in $a^+$, say $a^k$ with $0<k \le n$.
Thus, the word processed by $\rho_s$ has the form
\[
w_s = a(a^{n}b^{n})^{n_0} (a^{n-k} b^{n}) (a^{n}b^{n})^{n_1} (a^{n+k}b^{n}) (a^{n}b^{n})^{n_2} c^\omega
\]
for some $n_0, n_1, n_2 \ge 0$ with $n_0 + n_1 +n_2 + 2 = n+1$. 
Hence, $w_s$ contains a nonempty $(a,b)$-balanced prefix~$w'$ (ending in the $(n_0+1)$-th block of $b$'s).
Hence, $w_s \in L$.

We now show that the run~$\rho_s$ is not synchronously reachability accepting. 
As $\aut$ is deterministic, this implies that $w_s$ is not accepted by $(\aut,C)$, which yields the desired contradiction to $\Lsexists(\aut, C) = L$.

So, consider the run~$\rho_s$. 
Its prefix of length~$(n+1)2n+1$ processing
\[a(a^{n}b^{n})^{n_0} (a^{n-k} b^{n}) (a^{n}b^{n})^{n_1} (a^{n+k}b^{n}) (a^{n}b^{n})^{n_2}\] (which does not contain a $c$) cannot contain an $FC$-prefix.
Every such prefix could, due to completeness, be completed to an accepting run processing a word without a $c$, resulting in a contradiction.

So, consider a prefix of $\rho_s$ of length greater than $(n+1)2n+1$ and the prefix of $\rho$ of the same length. 
Both end in the same state (as the shifting is confined to the prefix of length $(n+1)2n+1$) and have the same extended Parikh image, as it only depends on the number of occurrences of transitions, not their order. 
As $\rho$ is synchronously reachability rejecting, we conclude that the prefix of $\rho_s$ is not an $FC$-prefix. 

So, $\rho_s$ is not synchronously reachability accepting, as it has no $FC$-prefix.
This yields the desired contradiction.

\ref{theorem_synchvsasynch_buchi}.)
We describe how to turn a given~\abpa~$(\aut, C)$ into an \sbpa~$(\aut', C')$ by synchronizing $F$-prefixes and $C$-prefixes (not necessarily all, infinitely many suffice). 
We do so by adding a flag~$f_{\text{acc}}$ to the state space that is deterministically raised when an accepting state is visited.
If this flag is high, it can nondeterministically be lowered. 
In the \sbpa we construct, only states reached by lowering the flag are accepting.
So, the flag should be lowered when the extended Parikh image of the current run prefix is in $C$, thereby synchronizing both events. 
Altogether, the flag we use can assume three distinct values: 
\myquot{low} (value $0$), \myquot{high} (value $1$), and \myquot{low, but high in the last step} (value $2$), i.e., lowered during the last transition.

Formally, consider an \abpa~$(\aut, C)$ with $\aut = (Q, \Sigma \times D, q_\init, \Delta, F)$.
We define~$\aut' = (Q \times \set{0,1,2}, \Sigma \times D, (q_\init, 0), \Delta', Q \times \set{2})$ where 
$\Delta'$ contains the following transitions for every $(q, (a, \vec{v}), q') \in \Delta$:
\begin{itemize}
    \item $((q,0),(a,\vec{v}),(q',\ind{F}(q')))$ where $\ind{F}\colon Q \rightarrow \set{0,1}$ is the indicator function for $F$: 
    If the flag is low, it is raised if and only if an accepting state is reached by the transition.
    
    \item $((q,1),(a,\vec{v}),(q',1))$: The flag stays high.
    
    \item $((q,1),(a,\vec{v}),(q',2))$: The flag is (nondeterministically) lowered.
    Note that we move to state~$(q',2)$ with the $2$ signifying that the flag was lowered during the last transition. 
    
    \item $((q,2),(a,\vec{v}),(q',\ind{F}(q')))$:  The flag has just been lowered and is raised again, if an accepting state reached by the transition. 
    
\end{itemize}
Then, we have $\LaBuchi(\aut, C) = \LsBuchi(\aut', C)$.

For the other direction, we show how to turn an \sbpa~$(\aut, C)$ into an equivalent \abpa. 
Note that we can assume without loss of generality that the initial state of $\aut$ is not accepting.
If it is, we just duplicate it as described in Item~\ref{theorem_synchvsasynch_reach} above to obtain an \nfa equivalent to $\aut$.

With this assumption, the construction turning an \srpa into an equivalent \arpa also works for the Büchi acceptance condition: we have $\LaBuchi(\aut', C') = \LsBuchi(\aut, C)$ where $(\aut', C')$ is obtained from $(\aut, C)$ by reflecting whether a run ends in an accepting state in the prefix's extended Parikh image. 

\ref{theorem_synchvsasynch_detbuchi}.)
Again, the transformation turning an \sbpa into an equivalent \abpa presented in Item~\ref{theorem_synchvsasynch_buchi} preserves determinism, which yields the inclusion. 

Now, we show the strictness of the inclusion.
Consider the language $L$ of infinite words over $\set{a,b,c}$ that contain 
infinitely many $(a,b)$-balanced prefixes and infinitely many $c$'s.
It is recognised by the deterministic \abpa~$(\aut, C)$ where $\aut$ is depicted in Figure~\ref{fig_infbalanced_and_infc} and $C = \set{(n,n) \mid n\in\nats}$.

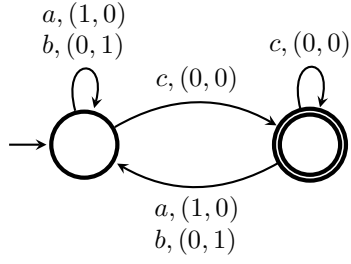
\begin{figure}
    \centering
    \begin{tikzpicture}[ultra thick]
    
    \node[state] (1) at (0,0) {};
    \node[state, accepting] (3) at (3,0) {};
    
    \path[-stealth]
    (-1,0) edge (1)
    (1) edge[loop above] node[align = left,above] {$a,(1,0)$\\$b,(0,1)$}()
    (1) edge[bend left] node[above] {$c,(0,0)$} (2)
    (2) edge[loop above] node[above] {$c,(0,0)$} ()
    (2) edge[bend left] node[below, align=left]{$a,(1,0)$\\$b,(0,1)$}(1)
    ;
    
    \end{tikzpicture}
    \caption{The automaton for Theorem~\ref{theorem_synchvsasynch}.\ref{theorem_synchvsasynch_detbuchi}.}
    \label{fig_infbalanced_and_infc}
\end{figure}

Now, towards a contradiction, assume 
there is a deterministic (w.l.o.g.\ complete) \sbpa~$(\aut,C)$ 
accepting $L$, say with $n$ states.
We need to introduce some language to simplify our arguments.
Let $\rho$ and $\rho'$ be two infinite runs or two run prefixes of the same length.
We say that $\rho$ covers $\rho'$ if the following holds:
for every $FC$-prefix of $\rho'$, the prefix of $\rho$  of the same length is also an $FC$-prefix.
If $\rho'$ is an infinite accepting run and $\rho$ covers $\rho'$, then $\rho$ is accepting as well.

We construct two infinite words~$w^-$, which is \emph{not} in $L$, and $w^+$, which is in $L$,
such that the (unique) run of~$\aut$ processing $w^-$ covers the (unique) run processing $w^+$.
This contradicts the fact that~$\aut$ accepts~$L$. 

The idea behind our construction is as follows:
we pick the word $w^-$ in the language~$a \big( (a^{n}b^{n})^* c \big)^\omega$,
which guarantees  $w^- \not\in L$
since it has no $(a,b)$-balanced prefix.
Then, we create the word $w^+$
by shifting cycles in the run of $(\aut,C)$ processing $w^-$,
in a way that creates infinitely many $(a,b)$-balanced prefixes.
The challenging part of the proof is that we need to guarantee that
shifting the cycles does not create new $FC$-prefixes.
By doing so, we ensure that the run processing $w^-$ covers the one processing $w^+$.

Formally, we define two sequences of run prefixes~$(\rho^-_i)_{i \in \mathbb{N}}$ and
$(\rho^+_i)_{i \in \mathbb{N}}$, processing finite words
$(w^-_i)_{i \in \mathbb{N}}$ and
$(w^+_i)_{i \in \mathbb{N}}$ respectively,
such that for all $i \ge 0$:
\begin{enumerate}
    \item\label{item:prefix}
    $\rho^-_{i}$ is a strict prefix of $\rho^-_{i+1}$ and $\rho^+_{i}$ is a strict prefix of $\rho^+_{i+1}$.
    \item\label{item:multiset}
    $\rho^-_i$ and $\rho^+_i$ have the same length, end in the same state, and have the same extended Parikh image.
    \item\label{item:accept}
    $\rho^-_i$ covers $\rho^+_i$.
    \item\label{item:unbalanced}
    $\occ{w^-_i}{a} = \occ{w^-_i}{b}+1$, $\occ{w^-_i}{c} = i$, and $w^-_i$ contains no $(a,b)$-balanced prefix.
    \item\label{item:balanced}
    $\occ{w^+_i}{a} = \occ{w^+_i}{b}+1$, $\occ{w^+_i}{c} = i$,
    and $w^+_i$ contains at least $i$ $(a,b)$-balanced prefixes.
\end{enumerate}
Then, we define $w^- = \lim_{i \to \infty} w^-_i$ and $w^+ = \lim_{i \to \infty} w^+_i$, which is well-defined due to Property~\ref{item:prefix}.
These words then satisfy the desired requirements:
$w^- \not\in L$ (due to Property~\ref{item:unbalanced}), $w^+ \in L$ (due to Property~\ref{item:balanced}),
and the run processing $w^-$ covers the run processing $w^+$ (due to determinism and Property~\ref{item:accept}).
Note that Property~\ref{item:multiset} is used in the construction below.

\newcommand{\greek}{\chi}

We define the sequences $(\rho^-_i)_{i \in \mathbb{N}}$
and $(\rho^+_i)_{i \in \mathbb{N}}$ inductively:
First, we set $\rho^-_0 = \rho^+_0$ to be the unique run prefix processing $a$, which satisfies all five conditions.
We now show how to build $\rho^-_{i+1}$ and $\rho^+_{i+1}$
based on $\rho^-_i$ and $\rho^+_i$
while preserving the desired properties.
Let $\greek^-$ denote the (unique) run of $(\aut,C)$
processing the infinite word~$w^-_i (a^{n}b^{n})^\omega$.
Since $\aut$ has $n$ states,
the run~$\greek^-$ visits a cycle
while processing each $a^{n}$-infix
occurring in $(a^{n}b^{n})^\omega$.
Therefore, one of these cycles, that we denote by $\pi$,
is visited infinitely often (recall that cycles are simple).

Hence, we can decompose $\greek^-$ into $\rho_i^- \greek_0   \pi  \greek_1   \pi   \greek_2   \pi   \cdots$, where we assume each $\greek_j$ to process at least one $b^n$ block (note that this can be achieved by adding some $ \pi$ to the $\greek_j$ if necessary).
Now, consider the run $\greek^+$ of $\aut$ obtained by
concatenating the run prefix $\rho_i^+$
with the infinite suffix $\greek_0   \pi  \greek_1   \pi   \greek_2   \pi   \cdots$,
and then swapping each occurrence of $\pi$
in the decomposition of the suffix
with the run infix~$\rho_j$ immediately following it:

\begin{tikzpicture}
\providecommand\x{}
\providecommand\y{}
\renewcommand{\x}{0.7}
\renewcommand{\y}{0.8}

\node[] at (-3.*\x,0.5*\y) {$\greek^-$ \strut};
\node[] at (-3.*\x,-0.5*\y) {$\greek^+$ \strut};

\node[] at (-1.75*\x,0.5*\y) {= \strut};
\node[] at (-1.75*\x,-0.5*\y) {= \strut};

\node[] at (-0.5*\x,0.5*\y) {$\rho_i^-$ \strut};
\node[] at (-0.5*\x,-0.5*\y) {$\rho_i^+$ \strut};

\node[] at (0.5*\x,0.5*\y) {$\greek_0$ \strut};
\node[] at (0.5*\x,-0.5*\y) {$\greek_0$ \strut};

\node[] at (14*\x,0.5*\y) {$\cdots$ \strut};
\node[] at (14*\x,-0.5*\y) {$\cdots$ \strut};


\foreach \i in
{1,2,3,4,5,6}{
\foreach \j in
{-1,1}{
\node[circle] at (2*\x*\i-0.5*\x*\j,-0.5*\y*\j) {$\greek_\i$ \strut};
\node[circle] at (2*\x*\i+0.5*\x*\j,-0.5*\y*\j) {$\pi$ \strut};
}}

\foreach \i in
{1,2,3,4,5,6}{
\draw[->,shorten <=2pt,rounded corners]
(2*\x*\i+0.3,0.3*\y) -- (2*\x*\i+0.2,-0.*\y) -- (2*\x*\i-0.2,-0.1*\y) -- (2*\x*\i-0.3,-0.4*\y);
\draw[->,line width=1mm,white,shorten <=2pt,rounded corners]
(2*\x*\i-0.3,0.3*\y) -- (2*\x*\i-0.2,-0.*\y) -- (2*\x*\i+0.2,-0.1*\y) -- (2*\x*\i+0.3,-0.4*\y);
\draw[->,shorten <=2pt,rounded corners]
(2*\x*\i-0.3,0.3*\y) -- (2*\x*\i-0.2,-0.*\y) -- (2*\x*\i+0.2,-0.1*\y) -- (2*\x*\i+0.3,-0.4*\y);
}
\end{tikzpicture}

The word processed by $\greek^+$ is not in $L$
as it contains finitely many $c$'s (only the finite prefix processed by~$\rho_i^+$ can contain $c$'s).
Therefore, it is not accepting.
As a consequence, there exists some integer $j \geq 0$
such that no $FC$-prefix ends in the infix~$\greek_j  \pi$ of $\greek^+$.

We are now ready to define $\rho^-_{i+1}$ and $\rho^+_{i+1}$.
Recall that $\rho_{j+1}$ processes a word of the form~$u b^n v$. 
Now, let $\greek_{j+1}'$ be the unique run prefix starting in the same state as $\greek_{j+1}$ processing $ub^n c$.

We define $\rho_{i+1}^-$ and $\rho_{i+1}^+$ as follows:

\medskip

\begin{tikzpicture}
\providecommand\x{}
\providecommand\y{}
\renewcommand{\x}{0.7}
\renewcommand{\y}{0.8}
\node[] at (-3.*\x,0.5*\y) {$\rho_{i+1}^-$ \strut};
\node[] at (-3.*\x,-0.5*\y) {$\rho_{i+1}^+$ \strut};

\node[] at (-1.75*\x,0.5*\y) {= \strut};
\node[] at (-1.75*\x,-0.5*\y) {= \strut};

\node[] at (-0.5*\x,0.5*\y) {$\rho_i^-$ \strut};
\node[] at (-0.5*\x,-0.5*\y) {$\rho_i^+$ \strut};

\node[] at (0.5*\x,0.5*\y) {$\greek_0$ \strut};
\node[] at (0.5*\x,-0.5*\y) {$\greek_0$ \strut};

\node[] at (6*\x,0.5*\y) {$\cdots$ \strut};
\node[] at (6*\x,-0.5*\y) {$\cdots$ \strut};

\node[] at (7.5*\x,0.5*\y) {$\pi$ \strut};
\node[] at (7.5*\x,-0.5*\y) {$\pi$ \strut};

\node[] at (8.5*\x,0.5*\y) {$\greek_{j-1}$ \strut};
\node[] at (8.5*\x,-0.5*\y) {$\greek_{j-1}$ \strut};

\node[] at (11.5*\x,0.5*\y) {$\pi$ \strut};
\node[] at (11.5*\x,-0.5*\y) {$\pi$ \strut};

\node[] at (12.5*\x,0.5*\y) {$\greek_{j+1}'$ \strut};
\node[] at (12.5*\x,-0.5*\y) {$\greek_{j+1}'$ \strut};


\foreach \i in
{1,2}{
\foreach \j in
{-1,1}{
\node[circle] at (2*\x*\i+0.5*\x,-0.5*\y*\j) {$\greek_\i$ \strut};
\node[circle] at (2*\x*\i-0.5*\x,-0.5*\y*\j) {$\pi$ \strut};
}}

\foreach \i in
{5}{
\foreach \j in
{-1,1}{
\node[circle] at (2*\x*\i-0.5*\x*\j,-0.5*\y*\j) {$\greek_j$ \strut};
\node[circle] at (2*\x*\i+0.5*\x*\j,-0.5*\y*\j) {$\pi$ \strut};
}}

\foreach \i in
{5}{
\draw[->,shorten <=2pt,rounded corners]
(2*\x*\i+0.3,0.3*\y) -- (2*\x*\i+0.2,-0.*\y) -- (2*\x*\i-0.2,-0.1*\y) -- (2*\x*\i-0.3,-0.4*\y);
\draw[->,line width=1mm,white,shorten <=2pt,rounded corners]
(2*\x*\i-0.3,0.3*\y) -- (2*\x*\i-0.2,-0.*\y) -- (2*\x*\i+0.2,-0.1*\y) -- (2*\x*\i+0.3,-0.4*\y);
\draw[->,shorten <=2pt,rounded corners]
(2*\x*\i-0.3,0.3*\y) -- (2*\x*\i-0.2,-0.*\y) -- (2*\x*\i+0.2,-0.1*\y) -- (2*\x*\i+0.3,-0.4*\y);
}
\end{tikzpicture}

Note that the suffix of $\rho_{i+1}^+$ after $\rho_i^+$ differs from the suffix of $\rho_{i+1}^-$ after $\rho_i^-$ 
by the swap of $\rho_j$ with the cycle~$ \pi$ preceding it.
Also, note that this is one of the swaps used to obtain the suffix of
$\greek^+$ from the suffix of $\greek^-$,
and it does not create a new $FC$-prefix in $\greek_j\pi$, by the choice of $j$.

Note that $\greek_0  \pi \greek_1 \cdots  \greek_j   \pi  \greek_{j+1}'$
processes a word in $(a^{n}b^{n})^+ c$.
Thus, Properties~$\ref{item:prefix}$ and $\ref{item:unbalanced}$ 
follow immediately from the definition of $\rho_{i+1}^-$ and $\rho_{i+1}^+$ and the induction hypothesis.
Furthermore, Property~$\ref{item:multiset}$ follows from Remark~\ref{remark:shifting} and the induction hypothesis.

Now, consider Property~\ref{item:balanced}:
To see that $w^+_{i+1}$ contains (at least) one more $(a,b)$-balanced prefix than $w^+_i$,
note that, in the run $\rho_{i+1}^-$, after each $b^n$ block
the difference between the number of $a$'s and $b$'s processed
so far is equal to $1$.
Therefore, since in the run $\rho_{i+1}^+$
we swapped one of the cycles~$ \pi$ (processing a nonzero number of $a$'s)
with $\greek_j$ (processing at least one $b^n$ block),
we created a new $(a,b)$-balanced prefix.
All other requirements of Property~\ref{item:balanced} follow from the induction hypothesis and arguments similar to those for Property~\ref{item:unbalanced}.

Finally, consider Property~\ref{item:accept}, i.e., we need to show that for every $FC$-prefix of $\rho_{i+1}^+$, the prefix of $\rho_{i+1}^-$ of the same length is also an $FC$-prefix.
We proceed by case distinction.
So, consider a prefix of $\rho_{i+1}^+$.
\begin{itemize}
    \item If it is even a prefix of $\rho_{i}^+$, then we can apply the induction hypothesis.
    \item 
If it ends in the part~$\greek_0   \pi  \greek_1   \pi  \cdots   \pi   \greek_{j-1}$, then the corresponding prefix of $\rho_{i+1}^-$ ends in the same state and has the same extended Parikh image by induction hypothesis (Property~\ref{item:multiset}).
\item If it ends in the part~$\greek_j\pi$ then it is not an $FC$-prefix by the choice of $j$.
\item If it ends in the part~$\pi\greek_{j+1}'$, then the corresponding prefix of $\rho_{i+1}^-$ ends in the same state and has the same extended Parikh image due to Remark~\ref{remark:shifting} and the induction hypothesis.\qedhere
\end{itemize}
\end{proof}

\subsection{Proof of Theorem~\ref{theorem_detvsnondet}}

\begin{proof}
\ref{theorem_detvsnondet_reach}.)
Consider the language~$\doublebal$ of infinite words over $\set{a,b,c,d}$ that have prefixes~$w_1$ and $w_2$ such that $\occ{w_1}{a} = \occ{w_1}{b} >0$ and $ \occ{w_2}{c} =\occ{w_2}{d} > 0$. Note that $\occ{w_1}{a} = \occ{w_1}{b} >0$ is a stronger requirement than $w_1$ being nonempty and $(a,b)$-balanced: the single-letter word~$c$ is nonempty and $(a,b)$-balanced, but $\occ{c}{a} = \occ{c}{b}=0$.
It is accepted by the nondeterministic \srpa~$(\aut, C)$ with $\aut$ shown in Figure~\ref{fig_doublebalanced} and $C = \set{(n_0,n_1,n_2,n_3) \mid n_0 = n_1 >0 \text{ and } n_2 = n_3 >0}$.
Note that $\aut$ uses nondeterminism to freeze one pair of counters (thereby picking $w_1$ or $w_2$). 

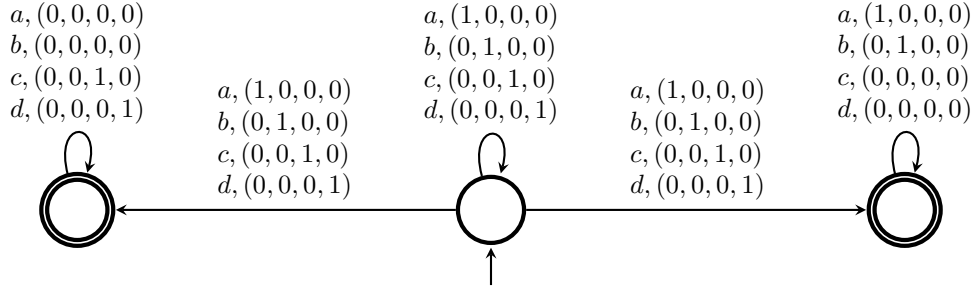
\begin{figure}
    \centering
    \begin{tikzpicture}[ultra thick]
    
    \node[state] (1) at (0,0) {};
    \node[state, accepting] (2) at (5.5,0) {};
    \node[state, accepting] (3) at (-5.5,0) {};
    
    \path[-stealth]
    (0,-1) edge (1)
    (1) edge[loop above] node[align = left,above] {$a,(1,0,0,0)$\\$b,(0,1,0,0)$\\$c,(0,0,1,0)$\\$d,(0,0,0,1)$}()
    (1) edge[] node[align=left,above] {$a,(1,0,0,0)$\\$b,(0,1,0,0)$\\$c,(0,0,1,0)$\\$d,(0,0,0,1)$}(2)
    (1) edge[] node[align=left,above] {$a,(1,0,0,0)$\\$b,(0,1,0,0)$\\$c,(0,0,1,0)$\\$d,(0,0,0,1)$}(3)
    (2) edge[loop above] node[align = left,above] {$a,(1,0,0,0)$\\$b,(0,1,0,0)$\\$c,(0,0,0,0)$\\$d,(0,0,0,0)$}()
    (3) edge[loop above] node[align = left,above] {$a,(0,0,0,0)$\\$b,(0,0,0,0)$\\$c,(0,0,1,0)$\\$d,(0,0,0,1)$}()
    ;
    
    \end{tikzpicture}
    \caption{The automaton for Lemma~\ref{theorem_detvsnondet}.\ref{theorem_detvsnondet_reach}.}
    \label{fig_doublebalanced}
\end{figure}

Now, towards a contradiction, assume that $\doublebal$ is accepted by some deterministic \arpa~$(\aut,C)$, say with $n$ states.
Now, consider $w = a (a^nb^n)^{n+1}(c^n d^n)^\omega \notin L$.
Due to completeness, $\aut$ has a run~$\rho$ processing $w$.
Every infix of $\rho$ processing $b^n$ contains a cycle. 
Hence, there are two such cycles that start in the same state.
Let $\rho_s$ be the run obtained by shifting the second cycle to the end of the first cycle, and let $w_s$ be the word processed by $\rho_s$.
Shifting the nonempty cycle creates a prefix~$w_1$ with $\occ{w_1}{a} = \occ{w_1}{b} >0$. 
Furthermore, as the shifting is restricted to the prefix containing the $a$'s and $b$'s, $w_s$ also has a prefix~$w_2$ with $\occ{w_2}{c} = \occ{w_2}{d} >0$, e.g., the one ending with processing the first~$c^n d^n$ infix.
Hence, $w_s \in \doublebal$ and $\rho_s$ is accepting due to determinism.

We show that the run~$\rho$ is rejecting, yielding the desired contradiction.
First, let us remark that exactly the same states occur in $\rho$ and $\rho_s$.
Hence, as $\rho_s$ contains an accepting state, so does $\rho$.
Also, there is a $C$-prefix~$\rho'$ of $\rho_s$.

First, we consider the case where the prefix~$\rho'$ has length at most $(n+1)2n+1$, i.e., it processes a word containing only $a$'s and $b$'s.
Each run infix of $\rho_s$ that processes $c^n$ contains a cycle. 
We pump this cycle once in the first such infix. 
This yields an accepting run, as the $C$-prefix~$\rho'$ is preserved (it appears before the pumping position) and the set of states occurring is unchanged by the pumping.
However, the word processed by the resulting run does not even have a $(c,d)$-balanced prefix, as the first $c$-block now has more than $n$ $c$'s.
So, we have derived the desired contradiction in this case.

Now, consider the case where the $C$-prefix~$\rho'$ of $\rho_s$ has length greater than $(n+1)2n+1$, i.e., it processes at least one $c$.
Then, as the shifting used to obtain $\rho_s$ from $\rho$ is confined to the prefix of length $(n+1)2n+1$, we conclude that the prefix of $\rho$ of length~$\size{\rho'}$ has the same extended Parikh image as $\rho'$.
Hence, $\rho$ is also accepting, yielding again the desired contradiction.

\ref{theorem_detvsnondet_safety}.)
Next, we consider safety acceptance.
Let $L = L' \cup \set{a,\$}^\omega$ with
\[
L' = \set{ a^{n_0} \$ a^{n_1} \$ \cdots \$ a^{n_{k}} \$ b^n \$^\omega \mid\text{ $k >0$ and $n_i > n$ for some $0 \le i \le k$ }},
\]
which is accepted by the~\spa~$(\aut, C)$ with $\aut$ in Figure~\ref{figure_safetysep} and $C = \set{(n,n') \mid n > n'} \cup\set{(0,0)}$.
Note that adding $\set{a,\$}^\omega$ ensures that $L$ can be accepted by a safety automaton. The language~$L'$ itself cannot be accepted by a safety automaton, as it requires the occurrence of a $\$$ after an unbounded number of $a$'s (stated differently, it is not a closed set in the Cantor topology).

\begin{figure}
    \centering
    \begin{tikzpicture}[ultra thick]
    
    \node[state,accepting] (1) at (0,0) {};
    \node[state,accepting] (2) at (3,0) {};
    \node[state,accepting] (3) at (6,0) {};
    \node[state,accepting] (4) at (9,0) {};
    \node[state,accepting] (5) at (12,0) {};

    \path[-stealth]
    (-1,0) edge (1)
    (1) edge[loop above] node[above, align=left] {$a,(0,0)$\\$\$,(0,0)$} ()
    (1) edge node[above,align=left] {$a,(1,0)$\\$\$,(0,0)$} (2)
    (2) edge[loop above] node[above, align=left] {$a,(1,0)$} ()
    (2) edge node[above] {$\$,(0,0)$} (3)
    (3) edge[loop above] node[above, align=left] {$a,(0,0)$\\$\$,(0,0)$} ()
    (3) edge node[above] {$\$,(0,0)$} (4)
    (4) edge[loop above] node[above, align=left] {$b,(0,1)$} ()
    (4) edge node[above] {$\$,(0,0)$} (5)
    (5) edge[loop above] node[above, align=left] {$\$,(0,0)$} ()
    (2) edge[bend right] node[below] {$\$,(0,0)$} (4)
    ;
    
    \end{tikzpicture}
    \caption{The automaton for Lemma~\ref{theorem_detvsnondet}.\ref{theorem_detvsnondet_safety}.}
    \label{figure_safetysep}
\end{figure}
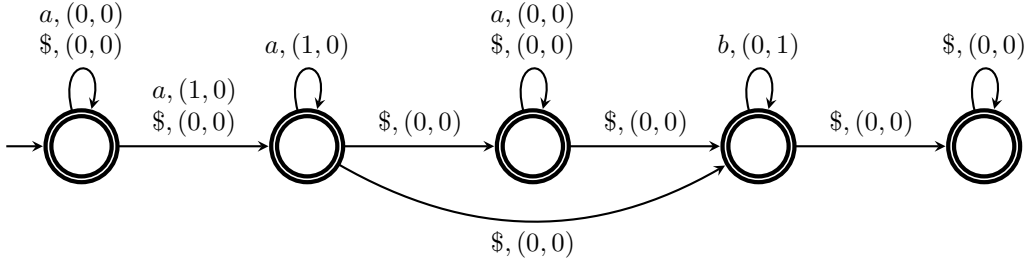

Towards a contradiction, assume that $L$ is accepted by a deterministic~\spa~$(\aut, C)$, say with $n$ states. 
Due to Remark~\ref{remark:completeness}, we can assume the automaton to be complete.
Now, consider the word~$w = (a^n\$)^{n+1}b^n\$^\omega \notin L$.
Due to completeness, the automaton has a (unique) run~$\rho$ processing~$w$, which is rejecting.
Hence, there is a prefix~$\rho'$ of $\rho$ that either ends in a nonaccepting state or whose extended Parikh image is not in $C$.

First, we show that $\rho'$ has at least length~$(n+1)^2+1$.
If not, then $\rho'$ does not process a single~$b$, which means we can extend it with a suffix processing~$a^\omega$ (recall that $\aut$ is complete). 
The resulting run is rejecting, due to the the prefix~$\rho'$, but processes a word in $(a^* \$)^*a^\omega \subseteq L$, a contradiction.

Every infix of $\rho$ processing an infix of the form~$a^n$ contains a cycle.
Hence, there are two such cycles starting in the same state.
By shifting the first of these cycles in front of the second one we obtain a new run~$\rho_s$ processing a word of the form~$(a^n\$)^{n_0} a^{n-k}\$ (a^n\$)^{n_1}a^{n+k}\$ (a^n\$)^{n_2} b^n\$^\omega$, where $k>0$ is the length of the first cycle.
This word is in $L$, but the unique run~$\rho_s$ of $(\aut, C)$ processing it is rejecting: 
its prefix of length~$\size{\rho'}$ ends in the same state as $\rho'$ and has the same extended Parikh image as $\rho'$, as the shifting is restricted to this prefix. 
This again yields a contradiction.

\ref{theorem_detvsnondet_buchi}. and \ref{theorem_detvsnondet_cobuchi}.)
To conclude, we simultaneously consider Büchi and co-Büchi acceptance. 
Let $\bal$ be the language of infinite words that have a nonempty $(a,b)$-balanced prefix.
The deterministic~\srpa~$(\aut, C)$ with $\aut$ as in Figure~\ref{figure_balanced} and $C = \set{(n,n) \mid n > 0}$ accepts $\bal$.
In Theorem~\ref{thm_omegaseparations}, we show that every \rpa can be turned into an equivalent \bpa and into an equivalent \cpa.
Thus, it remains to show that $\bal$ is not accepted by the deterministic variants.

Towards a contradiction, assume that $\bal$ is accepted by a deterministic~\bpa or by a deterministic~\cpa, say with $n$ states.
Due to Remark~\ref{remark:completeness}, we can assume the automaton to be complete.

Consider $w = a(a^nb^n)^\omega \notin L$, i.e, $(\aut, C)$ has a nonaccepting run~$\rho$ processing~$w$.
Here, we consider cycles in $\rho$ that process the $b^n$-infixes, starting in the same state. 
We shift one of them forward, and obtain a run~$\rho_s$ processing a word with a nonempty $(a,b)$-balanced prefix, i.e., the word processed by $\rho_s$ is in $\bal$.
However, its (unique due to determinism) run~$\rho_s$ is not accepting, as both asynchronous and synchronous Büchi nonacceptance as well as co-Büchi nonacceptance are preserved under shifting a single cycle. 
\end{proof}

\subsection{Proof of Theorem~\ref{thm_omegaseparations}}

\begin{proof} \textbf{\rpa \boldmath$\subseteq$ \bpa and \rpa \boldmath$\subseteq$ \cpa})
Let $(\aut,C)$ be an \rpa and recall that $\aut$ is complete by assumption.
Without loss of generality, we assume that $(\aut, C)$ is an \srpa (recall Theorem~\ref{theorem_synchvsasynch}).
Hence, a word~$w$ is in $\Lsexists(\aut, C)$ if and only if there is an $FC$-prefix processing a prefix of $w$. 
Due to completeness, we can always extend such a run prefix into an infinite accepting run.

By \myquot{freezing} (see the proof of Theorem~\ref{theorem_synchvsasynch}.\ref{theorem_synchvsasynch_detreach}) the extended Parikh image after such a prefix, we ensure that almost all run prefixes are $C$-prefixes. 
To implement the freezing, we add a sink state~$q_s$ equipped with a self-loop labeled with a zero vector and add transitions leading from every accepting state to the sink, again labeled with a zero vector.
Formally, given $\aut = (Q, \Sigma \times D, q_\init, \Delta, F)$, we define $\aut' = (Q \cup \set{q_s}, \Sigma \times (D \cup \set{\vec{0}}), q_\init, \Delta \cup \Delta', \set{q_s})$ where $\vec{0}$ is the zero vector of appropriate dimension and 
\[\Delta' = \set{(q, (a, \vec{0}), q_s) \mid q \in F \text{ and } a \in \Sigma} \cup \set{ (q_s,(a,\vec{0}),q_s) \mid a \in \Sigma }.\]
So, from every state in $F$, we can nondeterministically transition to the sink state (which should only be done if the extended Parikh image is in $C$).
Then, we have $\Lsexists(\aut, C) = \LaBuchi(\aut', C) = \LsBuchi(\aut', C) = \LcoBuchi(\aut', C)$.

Now, we show that all other automata types are pairwise incomparable.
Due to the sheer number of cases, and the fact that most proofs rely on the shifting property introduced in Remark~\ref{remark:shifting} and are similar to the arguments described in detail in the proofs of Theorem~\ref{theorem_synchvsasynch} and Theorem~\ref{theorem_detvsnondet}, we only sketch them here.

We begin by showing that the two inclusions proved above are as tight as possible, e.g., an \rpa can be turned into a nondeterministic \bpa and into a nondeterministic \cpa, but in general not into a deterministic \bpa and not into a deterministic \cpa. 
Then, we prove the remaining non-inclusions.

\textbf{\rpa \boldmath$\not\subseteq$ deterministic \abpa}) We have shown in the proof of Theorem~\ref{theorem_detvsnondet} that the language~$\bal$ of infinite words containing a nonempty $(a,b)$-balanced prefix is accepted by a deterministic \srpa. 
Here, we show that it is not accepted by any deterministic \abpa.

Assume $\bal$ is accepted by a deterministic \abpa~$(\aut, C)$, say with $n$ states.
Due to Remark~\ref{remark:completeness} we assume that $\aut$ is complete.
The word~$w = a(a^nb^n)^\omega$ is not in $\bal$, i.e., the unique run~$\rho$ of $(\aut,C)$ processing $w$ is not accepting. 
However, the run infixes processing the infixes~$b^n$ each contain a cycle. 
So, we can find two cycles starting in the same state and shift the second one forward while preserving asynchronous Büchi acceptance of the run since the shift does not affect Büchi acceptance.
But the resulting run processes a word~$w_s$ with a nonempty $(a,b)$-balanced prefix, yielding the desired contradiction, as the unique run of $(\aut, C)$ processing $w_s$ has to be accepting.

\textbf{\rpa \boldmath$\not\subseteq$ deterministic \cpa}) The argument just presented showing that $\bal$ is not accepted by any deterministic \bpa also applies to deterministic \cpa, as co-Büchi acceptance is also preserved by shifting one cycle.

\textbf{\spa \boldmath$\not \subseteq$ \bpa and \spa \boldmath$\not \subseteq$ \cpa}) We show that the language
\[
\unbal = \set{w \in \set{a,b}^\omega \mid \text{all nonempty prefixes of $w$ are $(a,b)$-unbalanced}}
\]
is accepted by a deterministic \spa, but not by any \bpa nor by any \cpa. 
The \spa accepting the language is~$(\aut, C)$ where $\aut$ is depicted in Figure~\ref{figure_balanced} and $C = \set{(n,n') \mid n\neq n'} \cup \set{(0,0)}$.

Now, assume the language is accepted by some \bpa or \cpa, say with $n$ states.
We consider the word~$w = a (a^n b^n)^\omega \in \unbal$, which has an accepting run~$\rho$. 
For every infix~$b^n$ there is a cycle in the corresponding transitions of $\rho$.
Hence, we can find two such cycles starting in the same state. 
Shifting the second cycle to the front preserves both Büchi and co-Büchi acceptance, but the resulting run processes a word with a nonempty $(a,b)$-balanced prefix, yielding the desired contradiction.

\textbf{\bpa \boldmath$\not \subseteq$ \spa}) We show that the language~$\Infa$ of infinite words over~$\set{a,b}$ containing infinitely many $a$'s is accepted by a deterministic \sbpa, but not by any \spa. 
Constructing a deterministic \sbpa is trivial (the usual deterministic Büchi automaton with a universal semilinear set suffices), so we focus on the second part of the claim.

Assume $\Infa$ is accepted by an \spa~$(\aut, C)$. 
As $b^na^\omega$ is in $\Infa$ for every $n$, there is an accepting run~$\rho_n$ of $(\aut, C)$ for every word of this form. 
Hence, every prefix of each $\rho_n$ is an $FC$-prefix.
We arrange the prefixes of the $\rho_n$ processing the prefixes~$b^n$ in a finitely branching infinite tree.
Thus, K\H{o}nig's Lemma yields an infinite path through the tree.
By construction, this path is an accepting run of $(\aut, C)$ processing~$b^\omega \notin \Infa$, yielding the desired contradiction.

\textbf{\cpa \boldmath$\not\subseteq$ \spa}) We show that the language~$\Fina$ of infinite words over~$\set{a,b}$ containing finitely many $a$'s is accepted by a deterministic \cpa, but not by any \spa. 
Constructing a deterministic \cpa is trivial (the usual deterministic co-Büchi automaton with a universal semilinear set suffices), so we focus on the second part of the claim.

Assume $\Fina$ is accepted by an \spa~$(\aut, C)$. 
Analogously to the previous case, for every $a^nb^\omega \in \Fina$ there is an accepting run~$\rho_n$ of $(\aut, C)$.
The run prefixes processing the prefixes~$a^n$ can be arraigned in a tree and K\H{o}nig's Lemma yields an accepting run of $(\aut, C)$ processing~$a^\omega \notin \Fina$.

\textbf{\bpa \boldmath$\not \subseteq$ \rpa}) We have shown above that the language~$\Infa$ of infinite words with infinitely many $a$'s is accepted by a deterministic \sbpa.
Here, we show that it is not accepted by any \rpa.

Towards a contradiction, assume w.l.o.g.~(see Theorem~\ref{theorem_synchvsasynch}) it is accepted by an \srpa~$(\aut, C)$.
By definition, $\aut$ is complete.
Now, consider an accepting run~$\rho$ processing $a^\omega$.
There is an $FC$-prefix of $\rho$.
Due to completeness, we can extended this prefix to an accepting run of $(\aut, C)$ processing a word of the form~$a^n b^\omega$, yielding the desired contradiction.

\textbf{\cpa \boldmath$\not \subseteq$ \rpa}) We have shown above that the language~$\Fina$ of infinite words with finitely many $a$'s is accepted by a deterministic \sbpa.
Here, we show that it is not accepted by any \rpa.

Towards a contradiction, assume w.l.o.g.~(see Theorem~\ref{theorem_synchvsasynch}) it is accepted by an \srpa~$(\aut, C)$.
Analogously to the previous case, relying on completeness, we can extend a prefix of an accepting run processing $b^\omega$ into an accepting run of $(\aut, C)$ processing a word of the form~$b^n a^\omega$, yielding the desired contradiction.

\textbf{\rpa \boldmath$\not \subseteq$ \spa}) The language~$\Exa$ of infinite words over $\set{a,b}$ having at least one $a$ is accepted by a deterministic \srpa (the usual deterministic reachability automaton with a universal semilinear set suffices). 
Now, we show that $\Exa$ is not accepted by any \spa. 

So, towards a contradiction, assume that $\Exa$ is accepted by an \spa. 
The word~$b^na^\omega$ is in $\Exa$ for every $n \ge 0$, i.e., there is an accepting run processing each such word.
As before, we arrange the run prefixes processing the prefixes~$b^n$ in a finitely-branching infinite tree.
Then, K\H{o}nig's Lemma yields an accepting run processing~$b^\omega$, yielding the desired contradiction.

\textbf{\spa \boldmath$\not \subseteq$ \rpa}) We have shown above that the language~$\unbal$ is accepted by a deterministic \spa, but not by any \bpa. 
Thus, the fact that every \rpa can be effectively be turned into an equivalent \bpa implies that $\unbal$ is not accepted by any \rpa.

\textbf{\cpa \boldmath$\not\subseteq$ \bpa})
Here, we consider the language~$\almostallunbal \subseteq \set{a,b}^\omega$ of infinite words such that almost all prefixes are $(a,b)$-unbalanced. 
It is recognized by the deterministic~\cpa~$(\aut, C)$ where $\aut$ is depicted in Figure~\ref{figure_balanced} and $C = \set{(n,n') \mid n \neq n'}$.

Now, towards a contradiction, assume that $\almostallunbal$ is accepted by a nondeterministic~\sbpa~$(\aut, C)$, say with $n$ states.
We consider the word~$a(a^nb^n)^\omega$, which is processed by some accepting run~$\rho$.
As we consider synchronous acceptance, the set of $FC$-prefixes of $\rho$ is infinite.

We inductively define a sequence of runs~$\rho_j$ maintaining the following invariant:
each $\rho_j$ is accepting and processes a word of the form~$w_j (a^nb^n)^\omega$ where $w_j$ has at least $j$ nonempty $(a,b)$-balanced prefixes and $\occ{w_j}{a} = \occ{w_j}{b}+1$.

We start with $\rho_0 = \rho$, which satisfies the invariant with $w_0 = a$.
Now, consider $\rho_j$ for some $j \geqslant 0$ which processes~$w_j (a^nb^n)^\omega $.
Due to the invariant, there is some proper extension~$w_j'$ of $w_j$ such that the prefix of $\rho_j$ processing $w_j'$ is an $FC$-prefix.
After that prefix, we can find two cycles processing only $a$'s starting in the same state. 
Let $\rho_{j+1}$ be the run obtained by shifting the first one to the second one.
It processes a word of the form~$w_jw(a^nb^n)^\omega$ for some $w$ that contains the shifted infixes.
We define $w_{j+1} = w_jw$.

The run~$\rho_{j+1}$ is still accepting, as the shift only changes the extended Parikh image of finitely many prefixes. 
Furthermore, the requirement on the word~$w_{j+1}$ processed by $\rho_{j+1}$ is also satisfied, as we have introduced another $(a,b)$-balanced prefix by the shift, but the balance of almost all suffixes is left unchanged.

To conclude, notice that for every $j < j'$, $w_j$ is a strict prefix of $w_{j'}$ and the prefix of $\rho_j$ processing $w_j$ is also a prefix of $\rho_{j'}$, and contains at least $j$ $FC$-prefixes.
Hence, taking the limit of these prefixes yields an accepting run on a word with infinitely many $(a,b)$-balanced prefixes.
Hence, we have derived the desired contradiction. 

\textbf{\bpa \boldmath$\not\subseteq$ \cpa})
Recall that $\Infa$ is the language of infinite words over $\set{a,b}$ containing infinitely many $a$'s. 
It is accepted by a deterministic \sbpa obtained from the standard Büchi automaton accepting the language. 
Towards a contradiction, assume it is accepted by a \cpa~$(\aut, C)$. 

We inductively construct a sequence $(w_j)_{j\in\nats}$ of words with $w_{j+1} = w_j b^{n_{j+1}} a$ for some $n_{j+1} \in \nats$ and then show that every run of $(\aut, C)$ processing the limit of the $w_j$ is rejecting.
This yields the desired contradiction as the limit is in $\Infa$.
We construct $n_{j+1}$ such that every run processing $w_jb^{n_{j+1}}$ must have a non-$FC$-prefix while processing the suffix~$b^{n_{j+1}}$, relying on the fact that every run of the \cpa~$(\aut, C)$ must eventually have such a prefix when processing $w_j b^\omega \notin \Infa$ and on K\H{o}nig's Lemma.
This will ensure that every run processing the limit has infinitely many non-$FC$-prefixes. 
Before we start, we need to introduce some notation. 
We say that a run prefix~$\rho$ of $(\aut, C)$ processing some finite word~$w$ is bad (for co-Büchi acceptance), if it has a non-$FC$-prefix~$\rho'$ processing some word~$w'$ such that $w = w' b^n$ for some $n\ge 0$, i.e., $\rho'$ either ends in a nonaccepting state or has an extended Parikh image that is not in $C$ and afterwards $\rho$ only processes $b$'s.
Note that every non-bad prefix (one that is not bad) is an $FC$-prefix, as we allow $n = 0$. 
Note also that a run that has infinitely many bad prefixes and processes a word with infinitely many $a$'s is rejecting.

We begin the inductive definition with setting $w_0 = \epsilon$.
Now, assume we have already defined $w_j$ for some $j \ge 0$.
Let $T_{j+1}$ be the set of non-bad run prefixes of $(\aut, C)$ processing words of the form~$w_j b^n$ for some $n\ge 0$.
Note that $T_{j+1}$ is closed in the following sense: if a run prefix~$\rho$ processing $w_jb^n$ is in $T_{j+1}$ and $n' < n$, then the prefix of $\rho$ processing $w_jb^{n'}$ is also in $T_{j+1}$.

We argue that $T_{j+1}$ is finite. 
Towards a contradiction, assume it is not.
Then, we can arrange the run prefixes in $T_{j+1}$ into an infinite finitely-branching tree.
Applying K\H{o}nig's Lemma yields an infinite path, which corresponds to an infinite run of $(\aut, C)$ processing the word~$w_jb^\omega$.
Due to the closure property, all prefixes of this run that are longer than $\size{w_j}$ are non-bad, as they are in $T_{j+1}$.
In particular, as mentioned above, all of these prefixes are $FC$-prefixes.
Hence, the run is co-Büchi accepting, but processes a word with finitely many $a$'s, yielding a contradiction.
Thus, $T_{j+1}$ is indeed finite and we can pick an $n_{j+1}$ such that each run prefix in $T_{j+1}$ processes a word of the form~$w_j b^n$ with $n < n_{j+1}$.
We define $w_{j+1} = w_jb^{n_{j+1}}a$.

Now, the limit~$w = b^{n_1}ab^{n_2}ab^{n_3}a\cdots$ of the $w_j$ contains infinitely many $a$'s and is therefore in $\Infa$.
To conclude the argument, we show that every run~$\rho$ of $(\aut, C)$ processing $w$ is rejecting, yielding the desired contradiction to $(\aut, C)$ accepting $\Infa$.

To this end, fix some $j \ge 0$.
The run prefix of $\rho$ processing $w_j b^{n_{j+1}}$ is, by the choice of $n_{j+1}$, not in $T_{j+1}$ and therefore bad.
Thus, $\rho$ has infinitely many bad prefixes and processes a word with infinitely many $a$.
Thus, it is, as argued above, not accepting.

\end{proof}

\section{Proofs omitted in Section~\ref{sec:closure}}

\subsection{Proof of Theorem~\ref{theorem_detcomplementation}}

\begin{proof}
\ref{theorem_detcomplementation_sreach2safety}.) Let $(\aut, C)$ be a deterministic~\srpa accepting $\Lsexists(\aut, C) \subseteq \Sigma^\omega$. 
We have described in the proof of Theorem~\ref{theorem_synchvsasynch}.\ref{theorem_synchvsasynch_reach} how to reflect in the extended Parikh image of a run prefix whether this prefix ends in an accepting state or not.
To this end, one adds a new component that is odd if and only if the prefix does end in an accepting state. 
Now, we have $\Sigma^\omega \setminus \Lsexists(\aut, C') = \Lall(\aut', C')$ where $\aut'$ is obtained from $\aut$ by making every state accepting and where
\[
C' = \set{(v_0, \ldots, v_{d-1}, f) \mid (v_0,\ldots,v_{d-1}) \notin C \text{ or } f \text{ even}},
\]
i.e., the safety automaton checks whether every prefix does not end in an accepting state (the last component of the extended Parikh image, reflecting that information, is even) or the original extended Parikh image is not in $C$.

\ref{theorem_detcomplementation_areach2safety}.) Let $(\aut, C)$ be a deterministic~\arpa accepting $\Laexists(\aut, C) \subseteq \Sigma^\omega$.
As before, we reflect whether a run ends in an accepting state in its extended Parikh image.
Now, we construct two \spa from $(\aut, C)$.
Given an input~$w$, the first one checks whether every prefix of the unique run~$\rho$ of $\aut$ processing~$w$ ends in a rejecting state and the second one checks that the extended Parikh image of every run prefix of $\rho$ is not in $C$.
Thus, the union of the languages accepted by these two \spa is equal to the complement of $L(\aut, C)$.
As \spa are closed under union (see Theorem~\ref{thm_omegaclosureprops}), we obtain the desired complement automaton.

Finally, the fact that deterministic \arpa are strictly more expressive than deterministic \srpa implies that not every deterministic \arpa can be complemented into a deterministic \spa: Otherwise, every deterministic \arpa~$(\aut, C)$ could be turned into a deterministic \spa for the complement, which can be turned into a deterministic \srpa for the original language of~$(\aut, C)$ (see the next item), which contradicts Theorem~\ref{theorem_synchvsasynch}.\ref{theorem_synchvsasynch_detreach}.

\ref{theorem_detcomplementation_safety2reach}.) Let $(\aut, C)$ be a deterministic~\spa accepting $\Lall(\aut, C) \subseteq \Sigma^\omega$. 
Again, we reflect whether a run prefix ends in an accepting state in its extended Parikh image and then turn $(\aut, C)$ into a deterministic \srpa that accepts if and only if there is a run prefix that does not end in an accepting state \emph{or} whose original extended Parikh image is not in $C$.

\ref{theorem_detcomplementation_sbuchi2cobuchi}.) 
The construction presented in Item~\ref{theorem_detcomplementation_sreach2safety} for deterministic \srpa also turns a deterministic~\sbpa into a deterministic \cpa accepting its complement.

\ref{theorem_detcomplementation_abuchi2cobuchi}.) 
The construction presented in Item~\ref{theorem_detcomplementation_areach2safety} for deterministic \arpa also turns a deterministic~\abpa into a \cpa accepting its complement.
The proof that nondeterminism might be required is also analogous. 

\ref{theorem_detcomplementation_cobuchi2buchi}.) 
The construction presented in Item~\ref{theorem_detcomplementation_safety2reach} for deterministic \spa also turns a deterministic~\cpa into a deterministic \sbpa accepting its complement.
\end{proof}

\subsection{Proof of Theorem~\ref{theorem_nondetcomplementation}}

\begin{proof}
\ref{theorem_nondetcomplementation_reach2safety}.)
Consider the language
\[L= \set{ w \in \set{a,b}^\omega \mid \text{there exists $n \ge 1$ such that $w$ contains the infix~$ba^nb$ twice} }.\]
It is accepted by the nondeterministic \srpa~$(\aut, C)$ where $\aut$ is depicted in Figure~\ref{fig_twice} and $C =\set{(n,n) \mid n\in\nats}$.

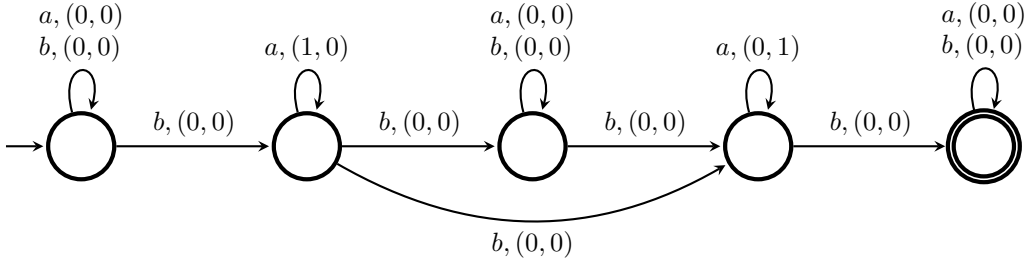
\begin{figure}
    \centering
    \begin{tikzpicture}[ultra thick]
    
    \node[state] (1) at (0,0) {};
    \node[state] (2) at (3,0) {};
    \node[state] (3) at (6,0) {};
    \node[state] (4) at (9,0) {};
    \node[state, accepting] (5) at (12,0) {};

    \path[-stealth]
    (-1,0) edge (1)
    (1) edge[loop above] node[align = left,above] {$a,(0,0)$\\$b,(0,0)$}()
    (2) edge[loop above] node[align = left,above] {$a,(1,0)$}()
    (3) edge[loop above] node[align = left,above] {$a,(0,0)$\\$b,(0,0)$}()
    (4) edge[loop above] node[align = left,above] {$a,(0,1)$}()
    (5) edge[loop above] node[align = left,above] {$a,(0,0)$\\$b,(0,0)$}()
    
    (1) edge[bend left=0] node[above] {$b,(0,0)$} (2)
    (2) edge[bend left=0] node[above] {$b,(0,0)$} (3)
    (3) edge[bend left=0] node[above] {$b,(0,0)$} (4)
    (4) edge[bend left=0] node[above] {$b,(0,0)$} (5)
    (2) edge[bend right] node[below] {$b,(0,0)$} (4)

    %
    ;
    
    \end{tikzpicture}
    \caption{The automaton for Theorem~\ref{theorem_nondetcomplementation}.\ref{theorem_nondetcomplementation_reach2safety}.}
    \label{fig_twice}
\end{figure}

Now, towards a contradiction, assume that the complement of $L$ is accepted by a nondeterministic \spa~$(\aut, C)$, say with $n$ states and $d$ counters.
Furthermore, let $m$ be the maximal entry of a vector labeling the transitions of $\aut$.

Then, every entry of the extended Parikh image of a run prefix of length~$\ell$ is bounded by $m\cdot \ell$.
Hence, the number of extended Parikh images reachable by run prefixes of length at most~$\ell$ is bounded by $(m\cdot \ell+1)^d$.

Given $x \ge 1$, define $\SET{x} = \set{1,\ldots, x}$.
Given a subset~$S$ of $\SET{x}$, we define the word~$w_S = ba^{s_1}ba^{s_2}b\cdots a^{s_{\size{S}}}$, where $s_1 < s_2 < \cdots < s_{\size{S}}$ is an enumeration of the elements of $S$.
Note that $\size{w_S}\le (x+1)^2$ for $S\subseteq \SET{x}$.
Also, $w_S w_{S'} b^\omega$ is in the complement of $L$ if and only if $S \cap S' =\emptyset$.

For $x \ge 1$ and $S\subseteq \SET{x}$, there is an accepting run~$\rho_S$ of $(\aut, C)$ processing the word~$w_S w_{\overline{S}}b^\omega$, where  $\overline{S} = \SET{x} \setminus S$.
We decompose $\rho_S$ into $\rho_S^1 \rho_S^2 \rho_S^3$ such that  $\rho_S^1$ processes $w_S$,
$\rho_S^2$ processes $w_{\overline{S}}$, and $\rho_S^3$ processes $b^\omega$.

Let $S \neq S'$ be two subsets of $\SET{x}$, i.e., there is (w.l.o.g.) some $s \in S$ that is not in $S'$.
We claim that $\rho_S^1$ and  $\rho_{S'}^1$ either end in different states or have distinct extended Parikh images.
If this is not the case, 
then $\rho_S^1 \rho_{S'}^2 \rho_{S'}^3$ is an accepting run of $(\aut, C)$ processing $w_S w_{\overline{S'}}b^\omega$ which is not in the complement of $L$, as $S \cap \overline{S'}$ contains $s$.

So, to accept the complement of $L$, the automaton~$(\aut, C)$ has to reach at least $2^x$ combinations of extended Parikh image and state for inputs of length at most $(x+1)^2$.
For large enough $x$, the quantity~$2^x$ exceeds the number of such combinations reachable by runs of length~$(x+1)^2$, which is at most $n\cdot (m\cdot(x+1)^2 +1)^d$.

Note that this proof also applies to Büchi and co-Büchi acceptance conditions, i.e., the complement of $L$ is accepted by none of the automata types we have introduced (the case of reachability being even simpler).

\ref{theorem_nondetcomplementation_safety2reach}.)
In the proof of Theorem~\ref{theorem_detvsnondet}.\ref{theorem_detvsnondet_safety}, we have shown that $L = L' \cup (a^* \$)^\omega \cup (a^* \$)^*a^\omega$ with
\[
L' = \set{ a^{n_0} \$ a^{n_1} \$ \cdots \$ a^{n_{k}} \$ b^n \$^\omega \mid \text{$k >0$ and $n < n_i$ for some $0 \le i \le k$ }}
\]
 is accepted by an \spa.
 Towards a contradiction, assume that its complement is accepted by an \rpa, say with $n$ states and, without loss of generality, with synchronous reachability acceptance.
 Note that the set
 \[
 \set{ a^{n_0} \$ a^{n_1} \$ \cdots \$ a^{n_{k}} \$ b^n \$^\omega \mid \text{ $k >0$ and $n \ge n_i$ for all $0 \le i \le k$ }}
 \]
 is a subset of $L$'s complement.
 In particular, it contains the word~$(a^n\$)^{n+1}b^n\$^\omega$. 
 So, consider an accepting run~$\rho$ of $(\aut, C)$ processing $w$, i.e., the run has an $FC$-prefix.
 This prefix has to be longer than $(n+1)^2$ as we can otherwise extend this $FC$-prefix into an accepting run processing the word~$(a^n\$)^\omega$, which is in $L$ and therefore is not accepted by $(\aut, C)$.
 
 Hence, the $FC$-prefix is longer than $(n+1)^2$.
 However, this allows us to shift some nonempty $a$-block around, thereby producing a run that is still accepting, but processing a word with an  $a$-block of length~${n+k}$ for some $k>0$, but only $n$ $b$'s. 
 This word is again in $L$
 and not in the complement, i.e., we have again derived a contradiction.

\ref{theorem_nondetcomplementation_buchi2cobuchi}.\ and \ref{theorem_nondetcomplementation_cobuchi2buchi}.) In the proof of Theorem~\ref{thm_omegaseparations}, we show that the language~$\bal$ of infinite words over $\set{a,b}$ that have a nonempty $(a,b)$-balanced prefix is accepted both by a \bpa and by a \cpa, as it is accepted by an \rpa.
However, in the same proof, we show that the complement of $\bal$, the set of words such that all nonempty prefixes are $(a,b)$-unbalanced is not accepted by any \bpa nor by any \cpa. 
\end{proof}

\subsection{Proof of Theorem~\ref{thm_omegaclosureprops}}

\begin{proof}
Throughout this proof, we fix two automata~$(\aut_i, C_i)$ for $i \in \set{1,2}$ with $\aut_i = (Q_i, \Sigma\times D_i, q_\init^i, \Delta_i, F_i)$, where we assume without loss of generality $Q_1 \cap Q_2=\emptyset$.
Furthermore, we assume without loss of generality that $C_1$ and $C_2$ have the same dimension, say $d$ (if this is not the case, we can apply Proposition~\ref{propsemilinearclosure}).
Finally, note that we also assume that both automata have the same alphabet. 
This is not a restriction, as we allow incomplete automata, even reachability ones, here.

Due to inclusions and techniques that work for several types of automata, we group the cases.

\textbf{Union for \rpa, \spa, \bpa, \cpa})
For all nondeterministic classes of automata, we prove closure under union by taking the disjoint union of two automata with a fresh initial state. 
We add a new dimension to the vectors labeling the transitions to ensure that the union of the two semilinear constraints is also disjoint.

Let $q_\init$ be a fresh state not in $Q_1 \cup Q_2$. 
Then, we define $(\aut, C)$ where $\aut = (Q_1 \cup Q_2 \cup \set{q_\init}, \Sigma \times D, q_\init, \Delta, F_1 \cup F_2)$ where $D = D_1 \cdot \set{(0),(1)} \cup D_2 \cdot \set{(0),(2)}$, and $\Delta$ is the union of the following sets of transitions.
\begin{itemize}
    \item $\set{ (q, (a, \vec{v}\cdot (0)),q') \mid (q, (a, \vec{v}),q') \in \Delta_1 \cup \Delta_2 }$: we keep all transitions of both automata, adding a zero in the last component.

    \item $\set{ (q_\init, (a, \vec{v}\cdot (i)),q') \mid (q_\init^i, (a, \vec{v}),q') \in \Delta_i \text{ for some }i }$: The new initial state has all transitions the initial states of the original automata have, adding an $i$ in the last component if the transition is copied from $\aut_i$.
\end{itemize}
So, after the first transition, the value in the last dimension is in $\set{1,2}$, is never updated, and reflects in which of the two automata the run proceeds. 
Thus, we define $C = C_1 \cdot \set{(1)} \cup C_2 \cdot \set{(2)}$ to take the disjoint union of the two semilinear sets.

Then, we have $\Laexists(\aut, C) = \Laexists(\aut_1, C_1) \cup \Laexists(\aut_2, C_2)$ as well as the corresponding statements for all other acceptance conditions (synchronous and asynchronous). 
Let us remark that this construction introduces nondeterminism, i.e., the initial choice in which automaton to proceed. 
We will later see that this is in some cases unavoidable.

\textbf{Union for deterministic \srpa, deterministic \sbpa})
For deterministic synchronous automata, we show that their union can be accepted by their product automaton, after reflecting whether a run prefix is an $F$-prefix in its extended Parikh image.

Recall the construction from the proof of Theorem~\ref{theorem_synchvsasynch}.\ref{theorem_synchvsasynch_reach} which reflects whether a nonempty run prefix ends in an accepting state in its extended Parikh image, i.e., there is a bijection between $FC$-prefixes in the original automaton and $C'$-prefixes in the new automaton.
Thus, we assume the $(\aut_i, C_i)$ to be of this form.
In particular, every state is accepting.

Now, we define the product automaton as follows:
$(\aut, C)$ where $\aut = (Q, \Sigma \times (D_1\cdot D_2),(q_\init^1,q_\init^2),\Delta, Q_1 \times Q_2)$ where
\[
((q_1,q_2),(a,\vec{v_1}\cdot\vec{v_2}),(q_1',q_2')) \in \Delta \text{ if and only if } (q_i, (a, \vec{v_i}),q_i')\in\Delta_i \text{ for } i\in\set{1,2},
\]
and $C = C_1 \cdot \nats^d \cup \nats^d \cdot C_2$.
Then, we have $\Lsexists(\aut, C) = \Lsexists(\aut_1,C_1) \cup \Lsexists(\aut_2,C_2)$ and $\LsBuchi(\aut, C) = \LsBuchi(\aut_1,C_1) \cup \LsBuchi(\aut_2,C_2)$ (for the Büchi case note that if a run of the product automaton has infinitely many $C$-prefixes, then one of the simulated runs has infinitely many $C_i$-prefixes and vice versa).

\textbf{Union for deterministic \arpa})
Let $L$ be the language of infinite words over $\set{a,b,c,d,e,f}$ containing a nonempty prefix~$w$ with $\occ{w}{a} = \occ{w}{b}$ and at least one $c$.
Similarly, let $L'$ be the language of infinite words over $\set{a,b,c,d,e,f}$ containing a nonempty prefix~$w$ with $\occ{w}{d} = \occ{w}{e}$ and at least one $f$.
Both languages are accepted by a deterministic \arpa (cp.\ the proof of Theorem~\ref{theorem_synchvsasynch}.\ref{theorem_synchvsasynch_detreach}), but we show that their union is not accepted by any deterministic \arpa.

Towards a contradiction, assume the union is accepted by a deterministic \arpa~$(\aut, C)$, say with $n$ states.
Consider the finite word~$w = ad(a^nb^n)^{n+1} (d^ne^n)^{n+1}$ and the unique run~$\rho$ of $\aut$ processing this prefix (which has to exist as $w$ can be extended to a word in the union).

As usual, we find a cycle in each run infix processing an infix~$a^n$ ($d^n)$.
Thus, we also find two cycles starting in the same state processing a word in $a^+$ ($d^+$).
Now, let $\rho_a$ ($\rho_d$) be the run obtained from $\rho$ by shifting the first of the $a$-cycles ($d$-cycles) to the back.
Note that $\rho_a$ ($\rho_d$) processes a word~$w_a$ ($w_d$) that has a nonempty prefix with the same number of $a$'s and $b$'s (the same number of $d$'s and $e$'s).
Also, $\rho$, $\rho_a$ and $\rho_d$ visit the same set of states, as we have just shifted cycles around.
In particular, if one of these runs has an $F$-prefix, then all have one.

Now consider the unique run~$\rho_a\rho_c$ of $\aut$ processing $w_ac^\omega$ (note that $\rho_c$ processes the whole $c^\omega$ suffix).
This is an extension of $\rho_a$ due to determinism of $\aut$. 
Furthermore, it is accepting as $w_ac^\omega$ is in $L$, i.e., it has an $F$-prefix and a $C$-prefix.
Note that not both of these prefixes can be prefixes of $\rho_a$, as we could otherwise construct an accepting run on a word without $c$'s and without $f$'s.
Similarly, it cannot be the case that the prefix~$\rho_a$ does contain neither an $F$-prefix nor a $C$-prefix: Otherwise, as both $\rho_a$ and $\rho$ end in the same state and have the same extended Parikh image, the run $\rho\rho_c$ is an accepting run on the word~$wc^\omega$ that is not in $L \cup L'$.
So, $\rho_a$ has either an $F$-prefix or a $C$-prefix, but not both.

An analogous argument shows that $\rho_d$ contains either an $F$-prefix or a $C$-prefix, but not both.
Now, recall that $\rho_a$ and $\rho_d$ contain the same states.
Hence, either both $\rho_a$ and $\rho_d$ contain an $F$-prefix but no $C$-prefix, or both contain a $C$-prefix but no $F$-prefix.

First, assume that they both contain an $F$-prefix. 
Then, the accepting run~$\rho_a\rho_c$ contains a $C$-prefix that is longer than $\rho_a$.
As $\rho$ and $\rho_a$ visit the same states, end in the same state, and have the same extended Parikh image, the run $\rho\rho_c$ is also accepting, but it processes the word~$w c^\omega$, which is not in $L \cup L'$, a contradiction.

Finally, assume that both $\rho_a$ and $\rho_d$ contain a $C$-prefix.
Then, the accepting run~$\rho_a\rho_c$ contains an $F$-prefix that is longer than $\rho_a$, i.e., $\rho_c$ contains an accepting state.
Then, $\rho_d\rho_c$ is also accepting, but it processes the word~$w_dc^\omega$ that is not in $L\cup L'$, a contradiction. 

\textbf{Union for deterministic \spa})
Let $L$ be the language of infinite words over $\set{a,b,c,d}$ such that all prefixes~$w$ of length at least two satisfy $\occ{w}{a} \neq \occ{w}{b}$.
Similarly, let $L'$ be the language of infinite words over $\set{a,b,c,d}$ such that all prefixes~$w$ of length at least two satisfy $\occ{w}{c} \neq \occ{w}{d}$.
Both are accepted by a deterministic \spa (cp.\ the proof of Theorem~\ref{thm_omegaseparations}).
We show that their union is not accepted by any deterministic \spa.

Towards a contradiction, assume it is accepted by some deterministic \spa~$(\aut, C)$, say with $n$ states.
Consider the word~$ac(a^nb^n)^{n+1}(c^n d^n)^{n+1}a^\omega \in L \cap L'$ and its unique accepting run~$\rho$ of $\aut$.
Hence, every prefix of $\rho$ is an $FC$-prefix.

In each infix of $\rho$ processing an infix~$a^n$ there is a cycle.
So, there are two starting with the same state. 
Moving the first of those to the back yields a run~$\rho_a$ on a word in $L' \setminus L$ (as it has a prefix of length at least two with the same number of $a$'s and $b$'s, but the $c$'s and $d$'s are not touched).
So, $\rho_a$ must still be accepting, i.e., every prefix of $\rho_a$ is an $FC$-prefix.

Dually, by shifting a cycle processing some $c$'s to the back, we obtain a run~$\rho_c$ on a word in $L \setminus L'$, so it must still be accepting, i.e., every prefix of $\rho_c$ is an $FC$-prefix.

However, due to determinism and the fact that both of these shifts are independent, we conclude that the run obtained by shifting both cycles is also accepting, as each its prefixes is a prefix of $\rho_a$ or $\rho_c$.
However, it processes a word that is neither in $L$ (as there is a prefix of length at least two with the same number of $a$'s and $b$'s obtained by shifting the $a$'s to the back) nor in $L'$ (as there is a prefix of length at least two with the same number of $c$'s and $d$'s obtained by shifting the $c$'s to the back).
This yields the desired contradiction.

\textbf{Union for deterministic \cpa})
The proof is a generalization of the previous one for deterministic \spa.
Consider the language~$L$ of infinite words over $\set{a,b,c,d}$ such that almost all prefixes~$w$ satisfy $\occ{w}{a} \neq \occ{w}{b}$.
Similarly, let $L$ be the language of infinite words over $\set{a,b,c,d}$ such that almost all prefixes~$w$ satisfy $\occ{w}{c} \neq \occ{w}{d}$.
Both are accepted by a deterministic \cpa (cp.\ the proof of Theorem~\ref{thm_omegaseparations}).
We show that their union is not accepted by any deterministic \cpa.

Here, we start with the word~$ac((a^nb^n)^{n+1}(c^n d^n)^{n+1})^\omega \in L \cap L'$.
Now, we can shift infinitely many cycles processing only $a$'s, thereby obtaining an accepting run~$\rho_a$ processing a word in $L'\setminus L$.
Dually, we can shift infinitely many cycles processing only $c$'s, thereby obtaining an accepting run~$\rho_c$ processing a word in $L\setminus L'$.
Now, doing both types of shifting infinitely often yields an accepting run, as due to determinism every prefix of the resulting run is a prefix of either $\rho_a$ or of $\rho_c$.
However, the resulting run processes a word with infinitely many prefixes that have the same number of $a$'s and $b$'s and with infinitely many prefixes that have the same number of $c$'s and $d$'s, which yields the desired contradiction.

\textbf{Intersection for \rpa})
Here, we assume the automata without loss of generality to be synchronous.
We again take the Cartesian product~$\aut$ of the two automata (taking the concatenation of the vectors), but allow the product automaton to nondeterministically freeze the counters of one of the automata during a transition leaving an  accepting state of that automaton (see the proof of Theorem~\ref{theorem_synchvsasynch}.\ref{theorem_synchvsasynch_detreach}). 
Then, using the constraint~$C = C_1 \cdot C_2$, we obtain $\Laexists(\aut, C) = \Lsexists(\aut_1, C_1) \cap \Lsexists(\aut_2, C_2)$. 

\textbf{Intersection for \spa, deterministic \spa, \cpa, deterministic \cpa})
Again, let $\aut$ be the Cartesian product of the two automata (without reflecting or freezing) and let $C = C_1 \cdot C_2$.
Note that this construction preserves determinism.
We have $\Lall(\aut, C) = \Lall(\aut_1, C_1) \cap \Lall(\aut_2, C_2)$ and $\LcoBuchi(\aut, C) = \LcoBuchi(\aut_1, C_1) \cap \LcoBuchi(\aut_2, C_2)$. 

\textbf{Intersection for deterministic \arpa})
We have shown in the proof of Theorem~\ref{theorem_detvsnondet}.\ref{theorem_detvsnondet_buchi} that the language~$\bal$ of infinite words over $\set{a,b}$ having a nonempty prefix~$w$ with $\occ{w}{a} = \occ{w}{b}$ is accepted by a deterministic \srpa, and thus also by a deterministic \arpa (Theorem~\ref{theorem_synchvsasynch}.\ref{theorem_synchvsasynch_detreach}).
Similarly, in the proof of Theorem~\ref{theorem_detvsnondet}.\ref{theorem_detvsnondet_reach} we have shown that the language of infinite words over $\set{a,b,c,d}$ that have prefixes $w_1, w_2$ with  $\occ{w_1}{a} = \occ{w_1}{b}$ and  $\occ{w_2}{c} = \occ{w_2}{d}$ is not accepted by any deterministic \arpa.
This yields the desired counterexample.

\textbf{Intersection for deterministic \srpa})
Recall that we have shown in Theorem~\ref{theorem_synchvsasynch}.\ref{theorem_synchvsasynch_detreach} that the language of infinite words over $\set{a,b,c}$ that have a nonempty $(a,b)$-balanced prefix and contain at least one $c$ is not accepted by any deterministic \srpa.
However both the language of infinite words containing an $(a,b)$-balanced prefix and the language of infinite words containing a $c$ are accepted by deterministic \srpa.
This yields the desired counterexample.

\textbf{Intersection for \bpa, deterministic \sbpa, deterministic \abpa})
In order to show that \bpa, deterministic \sbpa and deterministic \abpa
are not closed under intersection, we proceed as follows.
We know that both the language $L_{a,b} \subseteq \{a,b,c,d\}^\omega$ of words containing infinitely many $(a,b)$-balanced prefixes
and the language $L_{c,d} \subseteq \{a,b,c,d\}^\omega$
of words containing infinitely many  $(c,d)$-balanced prefixes
are accepted by deterministic \sbpa (thus also by deterministic \abpa and \bpa).
We show that the intersection $L_{a,b} \cap L_{c,d}$ is not accepted by a \bpa
(thus neither by a deterministic \abpa or \sbpa).

Towards a contradiction, assume there is a (non-deterministic) \bpa $(\aut,C)$
accepting $L_{a,b} \cap L_{c,d}$, say with $n$ states.
We show that, while the infinite word $w = a^{n}(c^{n}a^nb^{2n}a^nd^{n})^{\omega}$ is in $L_{a,b} \cap L_{c,d}$,
by swapping well chosen parts of some accepting run of $(\aut,C)$ processing $w$
we can build an accepting run of $(\aut,C)$ that processes a word that is not in $L_{a,b} \cap L_{c,d}$.

Let us consider an accepting run $\rho$ processing $w$.
Then either infinitely many $FC$-prefixes of $\rho$ end in the $a^nb^{2n}a^n$ blocks
of $w$,
or  infinitely many $FC$-prefixes of $\rho$ end in the $d^nc^n$ blocks of $w$.
We show how to reach a contradiction in the former case, and at the end of the proof we will explain
how the latter case can be treated similarly.

So, let us suppose that infinitely many $FC$-prefixes of $\rho$
end in the $a^nb^{2n}a^n$ blocks of the word $w = a^{n}(c^{n}a^nb^{2n}a^nd^{n})^{\omega}$.
Since $\aut$ has $n$ states, this means that
we can decompose $\rho$ into
$\rho_0 \greekC \greekABA \greekD \rho_1 \greekC \greekABA \greekD \rho_2 \greekC \greekABA \greekD \cdots$,
where the run infix~$\greekC \greekABA \greekD$
occurring infinitely often satisfies:
\begin{itemize}
    \item 
    $\greekC$ is a cycle processing a word in $c^+$.
    \item
    $\greekABA$ processes a word in $c^* a^n b^{2n} a^n d^*$,
    and there is an $FC$-prefix ending in each copy of $\greekABA$ in the decomposition above.
    \item
    $\greekD$ is a cycle processing a word in $d^+$.
\end{itemize}
Note that the $\rho_j$ may contain copies of $\greekC \greekABA \greekD$ where no $FC$-prefix ends.

We denote by $\rho'$ the run obtained by moving
all the odd copies of $\greekC$ (except the first) one step backward in this decomposition of $\rho$,
and all the odd copies of $\greekD$ one step forward:

\medskip
\begin{tikzpicture}
\providecommand\x{}
\providecommand\y{}
\renewcommand{\x}{0.8}
\renewcommand{\y}{1.0}

\node[] at (-1.25*\x,0.5*\y) {$\rho$ \strut};
\node[] at (-1.25*\x,-0.5*\y) {$\rho'$ \strut};

\node[] at (-0.5*\x,0.5*\y) {= \strut};
\node[] at (-0.5*\x,-0.5*\y) {= \strut};

\node[] at (0.25*\x,0.5*\y) {$\underline{\rho_0}$ \strut};
\node[] at (0.25*\x,-0.5*\y) {$\underline{\rho_0}$ \strut};

\node[] at (0.75*\x,0.5*\y) {$\underline{\greekC}$ \strut};
\node[] at (0.75*\x,-0.5*\y) {$\underline{\greekC}$ \strut};

\node[] at (13.25*\x,0.5*\y) {$\underline{\greekABA}$ \strut};
\node[] at (13.25*\x,-0.5*\y) {$\underline{\greekABA}$ \strut};

\node[] at (14*\x,0.5*\y) {$\cdots$ \strut};
\node[] at (14*\x,-0.5*\y) {$\cdots$ \strut};


\foreach \i in
{1,3,5}{
\node[circle] at (2*\x*\i-0.75*\x,0.5*\y) {$\underline{\greekABA}$ \strut};
\node[circle] at (2*\x*\i-0.25*\x,0.5*\y) {$\greekD$ \strut};
\node[circle] at (2*\x*\i+0.25*\x,0.5*\y) {$\rho_\i$ \strut};
\node[circle] at (2*\x*\i+0.75*\x,0.5*\y) {$\greekC$ \strut};
}

\foreach \i in
{2,4,6}{
\node[circle] at (2*\x*\i-0.75*\x,0.5*\y) {$\greekABA$ \strut};
\node[circle] at (2*\x*\i-0.25*\x,0.5*\y) {$\greekD$ \strut};
\node[circle] at (2*\x*\i+0.25*\x,0.5*\y) {$\rho_\i$ \strut};
\node[circle] at (2*\x*\i+0.75*\x,0.5*\y) {$\greekC$ \strut};
}

\foreach \i in
{1,3,5}{
\node[circle] at (2*\x*\i-0.75*\x,-0.5*\y) {$\underline{\greekABA}$ \strut};
\node[circle] at (2*\x*\i-0.25*\x,-0.5*\y) {$\rho_\i$ \strut};
\node[circle] at (2*\x*\i+0.25*\x,-0.5*\y) {$\greekC$ \strut};
\node[circle] at (2*\x*\i+0.75*\x,-0.5*\y) {$\greekC$ \strut};
}

\foreach \i in
{2,4,6}{
\node[circle] at (2*\x*\i-0.75*\x,-0.5*\y) {$\greekABA$ \strut};
\node[circle] at (2*\x*\i-0.25*\x,-0.5*\y) {$\greekD$ \strut};
\node[circle] at (2*\x*\i+0.25*\x,-0.5*\y) {$\greekD$ \strut};
\node[circle] at (2*\x*\i+0.75*\x,-0.5*\y) {$\rho_\i$ \strut};
}

\foreach \i in
{1,3,5}{
\draw[->,shorten <=2pt,rounded corners]
(2*\x*\i+2.7*\x,0.35*\y) -- (2*\x*\i+2.6*\x,0.1*\y) -- (2*\x*\i+0.3*\x,-0.1*\y) -- (2*\x*\i+0.2*\x,-0.35*\y);

\draw[->,line width=1mm,white,shorten <=2pt,rounded corners]
(2*\x*\i-0.3*\x,0.35*\y) -- (2*\x*\i-0.2*\x,0.1*\y) -- (2*\x*\i+2.1*\x,-0.1*\y) -- (2*\x*\i+2.2*\x,-0.35*\y);

\draw[->,shorten <=2pt,rounded corners]
(2*\x*\i-0.3*\x,0.35*\y) -- (2*\x*\i-0.2*\x,0.1*\y) -- (2*\x*\i+2.1*\x,-0.1*\y) -- (2*\x*\i+2.2*\x,-0.35*\y);
}
\end{tikzpicture}

Let $\balance{n}$ denote the difference between the number of $c$'s and the number of $d$'s occurring in the prefix of $\rho$ of length~$n$, and define $\balanceprime{n}$ similarly for prefixes of $\rho'$.
We now compare $\balance{n}$ and $\balanceprime{n}$ to argue that $\rho'$ processes a word that is not in $L_{c,d}$.

The underlined elements of $\rho$ and $\rho'$ denote the only parts containing positions~$n$ such that $\balance{n} = \balanceprime{n}$.
All the other positions~$n$ satisfy $\balanceprime{n} > \balance{n}$ since the length-$n$ prefix of $\rho'$ has either processed more $c$'s than the corresponding prefix of $\rho$ (as a $c$-cycle $\greekC$ has been moved backward)
or has processed less $d$'s (as a $d$-cycle $\greekD$ has been moved forward).
Thus, the word processed by $\rho'$ only has finitely many $(c,d)$-balanced prefixes.
This follows from the fact that along the run $\rho$
the difference~$\balance{n}$ between the number of $c$'s and $d$'s
processed so far never falls below $0$,
and is always strictly greater than $0$ along the run infixes~$\greekABA$.
Therefore, moving $c$-cycles backward and $d$-cycles forward in $\rho'$
ensures that, past the run prefix $\rho_0$,
the difference between the number of $c$'s and $d$'s
processed so far is always strictly greater than $0$. Thus, there is no $(c,d)$-balanced prefix.

Moreover, note that
along the underlined parts
both runs visit the same states and have the same extended Parikh images (see Remark~\ref{remark:shifting}).
Hence, the run $\rho'$ is accepting since
there is an $FC$-prefix ending in each
run infix $\greekABA$ of $\rho$.

This gives us the desired contradiction: the run $\rho'$ of $\aut$
is accepting, yet it processes a word that is not in $L_{a,b} \cap L_{c,d}$.

To conclude, if infinitely many $FC$-prefixes of $\rho$
end in the $d^nc^n$ blocks of $w$ (instead of the $a^nb^{2n}a^n$ blocks),
we again decompose $\rho$ into
$\rho_0 \greekA \greekABA \greekB \rho_1 \greekA \greekABA \greekB \rho_2 \greekA \greekABA \greekB \cdots$,
where this time:
\begin{itemize}
    \item 
    $\greekA$ is a cycle processing a word in $a^+$;
    \item
    $\greekABA$ processes a word in $a^* d^n c^n a^n b^{*}$,
    and there is an $FC$-prefix ending in each copy of $\greekABA$;
    \item
    $\greekB$ is a cycle processing a word in $b^+$.
\end{itemize}
We then complete the proof as in the previous case.

\textbf{Complement for \rpa, deterministic \srpa, deterministic \arpa})
The language of infinite words over~$\set{a,b}$ containing at least one $a$ is accepted by a deterministic \srpa. 
However, its complement, the language of infinite words having no $a$, is not accepted by any (even nondeterministic) \rpa:
Assume it is and consider an accepting run~$\rho$ processing $b^\omega$: it has a $F$-prefix and a $C$-prefix.
Due to our completeness assumption, after these prefixes the run can be continued in any way, e.g., by processing a word containing an $a$, while staying accepting, a contradiction.
This yields the desired counterexample for all three types of automata.

\textbf{Complement for \spa, deterministic \spa})
The language of infinite words over~$\set{a,b}$ containing no $a$ is accepted by a deterministic \spa. 
However, its complement, the language of infinite words having at least one $a$, is not accepted by any (even nondeterministic) \spa, as shown in the proof of Theorem~\ref{thm_omegaseparations}.
This yields the desired counterexample for both types of automata.

\textbf{Complement for \bpa, deterministic \sbpa, deterministic \abpa})
The language of infinite  words over $\set{a,b}$ containing infinitely many $(a,b)$-balanced prefixes is accepted by the deterministic \sbpa~$(\aut, C)$ where $\aut$ is depicted in Figure~\ref{figure_balanced} and $C = \set{(n,n) \mid n \ge 0}$.
However, we showed in the proof of Theorem~\ref{thm_omegaseparations} that its complement, the language~$U'$ of infinite words such that almost all prefixes are $(a,b)$-unbalanced, is not accepted by any (even nondeterministic) \bpa. 
This yields the desired counterexample for all three types of automata.

\textbf{Complement for \cpa})
The language of infinite words over $\set{a,b}$ containing a nonempty $(a,b)$-balanced prefix is accepted by a nondeterministic \cpa, as shown in the proof of Theorem~\ref{theorem_detvsnondet}.\ref{theorem_detvsnondet_buchi}.
However, we showed in the proof of Theorem~\ref{thm_omegaseparations} that its complement, the language~$U$ of infinite words that have only $(a,b)$-unbalanced nonempty prefixes, is not accepted by any nondeterministic \cpa. 
This yields the desired counterexample.

\textbf{Complement for deterministic \cpa}) Deterministic \cpa are closed under intersection, but not under union. 
Hence, due to De Morgan's laws, they cannot be closed under complementation.
\end{proof}







\section{Appendix: Parikh Automata and Two-counter Machines}
\label{sec_minsky}

Many of our undecidability proofs are reductions from nontermination problems for two-counter machines.
To simulate these machines by Parikh automata, we require them to be in a certain normal form. 
We first introduce two-counter machines, then the normal form, and conclude this section by presenting the simulation via Parikh automata. 

A two-counter machine~$\mach$ is a sequence
\[
(0:  \instr_0) (1:  \instr_1) \cdots (k-2:  \instr_{k-2})(k-1:  \stopp), 
\]
where the first element of a pair~$(\ell: \instr_\ell)$ is the line number and $\instr_\ell$ for $0 \le \ell < k-1$ is an instruction of the form
\begin{itemize}
    \item $\inc{i}$ with $i \in\set{0,1}$, 
    \item $\dec{i}$ with $i \in\set{0,1}$, or
    \item $\ite{i}{\ell'}{\ell''}$ with $i \in\set{0,1}$ and $\ell',\ell'' \in \set{0, \ldots,k-1}$. 
\end{itemize}
A configuration of $\mach$ is of the form~$(\ell, c_0, c_1)$ with $\ell \in \set{0, \ldots, k-1}$ (the current line number) and $c_0, c_1\in\nats$ (the current contents of the counters). 
The initial configuration is~$(0,0,0)$ and the unique successor configuration of a configuration~$(\ell,  c_0, c_1)$ is defined as follows:
\begin{itemize}
    \item If $\instr_\ell = \inc{i}$, then the successor configuration is $(\ell +1, c_0', c_1')$ with $c_i' = c_i +1$ and $c_{1-i}' = c_{1-i}$.
    \item If $\instr_\ell = \dec{i}$, then the successor configuration is $(\ell +1, c_0', c_1')$ with $c_i' = \max\set{c_i -1,0}$ and $c_{1-i}' = c_{1-i}$.
    \item If $\instr_\ell = \ite{i}{\ell'}{\ell''}$ and $c_i = 0$, then the successor configuration is $(\ell', c_0, c_1)$.
    \item If $\instr_\ell = \ite{i}{\ell'}{\ell''}$ and $c_i > 0$, then the successor configuration is $(\ell'', c_0, c_1)$.
    \item If $\instr_\ell = \stopp$, then $(\ell, c_0, c_1)$ has no successor configuration.
\end{itemize}
The unique run of $\mach$ (starting in the initial configuration) is defined as expected.
It is either finite (line~$k-1$ is reached) or infinite (line~$k-1$ is never reached).
In the former case, we say that $\mach$ terminates.

\begin{proposition}[\cite{Minsky67}]
The following problem is undecidable: Given a two-counter machine~$\mach$, does $\mach$ terminate?
\end{proposition}

In the following, we assume without loss of generality that each two-counter machine satisfies the \emph{guarded-decrement property}:
Every decrement instruction~$(\ell: \dec{i})$ is preceded by $(\ell-1: \ite{i}{\ell+1}{\ell})$ and decrements are never the target of a goto instruction. 
As the decrement of a zero counter has no effect, one can modify each two-counter machine~$\mach$ into an $\mach'$ satisfying the guarded-decrement property such that $\mach$ terminates if and only if $\mach'$ terminates: One just adds the the required guard before every decrement instruction and changes each target of a goto instruction that is a decrement instruction to the preceding guard.

The guarded-decrement property implies that decrements are only executed if the corresponding counter is nonzero. 
Thus, the value of counter~$i$ after a finite sequence of executed instructions (starting with value zero in the counters) is equal to the number of executed increments of counter~$i$ minus the number of executed decrements of counter~$i$. 
Note that the number of executed increments and decrements can be tracked by a Parikh automaton.

Consider a finite or infinite word~$w = w_0 w_1 w_2 \cdots$ over the set~$\set{0,1,\ldots, k-1}$ of line numbers.
We now describe how to characterize whether $w$ is (a prefix of) the projection to the line numbers of the unique run of $\mach$ starting in the initial configuration.
This characterization is designed to be checkable by a Parikh automaton.
Note that $w$ only contains line numbers, but does not encode values of the counters. These will be kept track of by the Parikh automaton by counting the number of increment and decrement instructions in the input, as explained above (this explains the need for the guarded-decrement property).
Formally, we say that $w$ contains an \emph{error} at position~$n < \size{w}-1$ if either $w_n = k-1$ (the instruction in line~$w_n$ is $\stopp$), or if one of the following two conditions is satisfied:
    \begin{enumerate}
        \item\label{casenonif} The instruction~$\instr_{w_n}$ in line~$w_{n}$ of $\mach$ is an increment or a decrement and $w_{n+1} \neq w_{n}+1$, i.e., the letter $w_{n+1}$ after $w_n$ is not equal to the line number $w_n + 1$, which it should be after an increment or decrement.
        
        \item\label{caseif} $\instr_{w_{n}}$ has the form~$\ite{i}{\ell}{\ell'}$, and one of the following cases holds: Either, we have
            \[
            \sum\nolimits_{j  \colon \instr_j = \inc{i}} \occ{w_0 \cdots w_{n}}{j} =
            \sum\nolimits_{j \colon \instr_j = \dec{i}} \occ{w_0 \cdots w_{n}}{j}
            \] and $w_{n+1} \neq \ell$,
            i.e., the number of increments of counter~$i$ is equal to the number of decrements of counter~$i$ in $w_0 \cdots w_{n}$ (i.e., the counter is zero) but the next line number in $w$ is not the target of the if-branch.
            Or, we have  
            \[
            \sum\nolimits_{j  \colon \instr_j = \inc{i}} \occ{w_0 \cdots w_{n}}{j} \neq 
            \sum\nolimits_{j \colon \instr_j = \dec{i}} \occ{w_0 \cdots w_{n}}{j} , 
            \]
            and $w_{n+1} \neq \ell'$,
            i.e., the number of increments of counter~$i$ is not equal to the number of decrements of counter~$i$ in $w_0 \cdots w_{n}$ (i.e., the counter is nonzero) but the next line number in $w$ is not the target of the else-branch.

    \end{enumerate}
Note that the definition of error (at position~$n$) refers to the number of increments and decrements in the prefix~$w_0 \cdots w_{n}$, which does not need to be error-free itself. 
However, if a sequence of line numbers does not have an error, then the guarded-decrement property yields the following result.

\begin{lemma}
\label{lemma-simulation}
Let $w \in \set{0,1,\ldots, k-1}^+$ with $w_0 = 0$. Then, $w$ has no errors at positions~$\set{0,1,\ldots, \size{w}-2}$ if and only if $w$ is a prefix of the projection to the line numbers of the run of $\mach$.
\end{lemma}

\begin{proof}
If $w$ has no errors at positions~$\set{0,1,\ldots, \size{w}-2}$, then an induction shows that $(w_n, c_0^n, c_1^n)$ with
\[c_i^n =\sum\nolimits_{j  \colon \instr_{j} = \inc{i}} \occ{w_0 \cdots w_{n-1}}{j} -
            \sum\nolimits_{j \colon \instr_{j} = \dec{i}} \occ{w_0 \cdots w_{n-1}}{j} \]
is the $n$-th configuration of the run of $\mach$.

On the other hand, projecting a prefix of the run of $\mach$ to the line numbers yields a word~$w$ without errors at positions~$\set{0,1,\ldots, \size{w}-2}$.
\end{proof}

The existence of an error can be captured by a Parikh automaton, leading to the undecidability of  the safe word problem for Parikh automata, which we now prove.
Let $(\aut, C)$ be a \pa accepting finite words over $\Sigma$. 
A \emph{safe} word of $(\aut, C)$ is an infinite word in $ \Sigma^\omega$ such that each of its prefixes is in $L(\aut, C)$. 

\begin{lemma}
\label{lemma_safeword}
The following problem is undecidable: Given a deterministic \pa, does it have a safe word? 
\end{lemma}

\begin{proof}
Our proof proceeds by a reduction from the nontermination problem for decrement-guarded two-counter machines.
Given such a machine~$\mach = (0:  \instr_0) \cdots (k-2:  \instr_{k-2})(k-1:  \stopp)$ let $\Sigma = \set{0, \ldots, k-1}$ be the set of its line numbers.
We construct a deterministic \pa~$(\aut_\mach, C_\mach)$ that accepts a word~$w \in \Sigma^*$ if and only if $w = \epsilon$, $w = 0$, or if $\size{w} \geqslant 2$ and $w$ does not contain an error at position~$\size{w}-2$ (but might contain errors at earlier positions).
Intuitively, the automaton checks whether the second-to-last instruction is executed properly. 
The following is then a direct consequence of Lemma~\ref{lemma-simulation}: $(\aut_\mach, C_\mach)$ has a safe word if and only if $\mach$ does not terminate.

The deterministic \pa~$(\aut_\mach, C_\mach)$ keeps track of the occurrence of line numbers with increment and decrement instructions of each counter (using four dimensions) and two auxiliary dimensions to ensure that the two cases in Condition~\ref{caseif} of the error definition on Page~\pageref{caseif} are only checked when the second-to-last letter corresponds to a goto instruction.
More formally, we construct $\aut_\mach$ such the unique run processing some input~$w = w_0\ldots w_{n-1}$ has the extended Parikh image~$(v_{\text{inc}}^0, v_{\text{dec}}^0, v_{\text{goto}}^0,v_{\text{inc}}^1, v_{\text{dec}}^1, v_{\text{goto}}^1)$ where
\begin{itemize}
    
    \item $v_{\text{inc}}^i$ is equal to $\sum_{j  \colon \instr_j = \inc{i}} \occ{w_0 \cdots w_{n-2}}{j}$, i.e., the number of increment instructions read so far (ignoring the last letter),
    
    \item $v_{\text{dec}}^i$ is equal to $\sum_{j \colon \instr_j = \dec{i}} \occ{w_0 \cdots w_{n-2}}{j}$, i.e., the number of decrement instructions read so far (ignoring the last letter), and
    
    \item $v_{\text{goto}}^i \bmod 4 = 0$, if the second-to-last instruction~$\instr_{w_{n-2}}$ is not a goto testing counter~$i$,
    
    \item $v_{\text{goto}}^i \bmod 4 = 1$, if the second-to-last instruction~ $\instr_{w_{n-2}}$ is a goto testing counter~$i$ and the last letter~$w_{n-1}$ is equal to the target of the if-branch of this instruction, and
    
    \item $v_{\text{goto}}^i \bmod 4 = 2$, if the second-to-last instruction~ $\instr_{w_{n-2}}$ is a goto testing counter~$i$ and the last letter~$w_{n-1}$ is equal to the target of the else-branch of this instruction.
    
    \item $v_{\text{goto}}^i \bmod 4 = 3$, if the second-to-last instruction~ $\instr_{w_{n-2}}$ is a goto testing counter~$i$ and the last letter~$w_{n-1}$ is neither equal to the target of the if-branch nor equal to the target of the else-branch of this instruction.
    Note that this constitutes an error at position~$n-2$.
    
\end{itemize}
Note that $v_{\text{goto}}^i \bmod 4 \neq 0$ can be true for at most one of the $i$ at any time (as a goto instruction~$\instr_{w_{n-2}}$ only refers to one counter) and that the $v_{\text{inc}}^i$ and $v_{\text{dec}}^i$ are updated with a delay of one transition (as the last letter of $w$ is ignored).
This requires to store  the previously processed letter in the state space of $\aut_\mach$.

Further, $C_\mach$ is defined such that $(v_{\text{inc}}^0, v_{\text{dec}}^0, v_{\text{goto}}^0,v_{\text{inc}}^1, v_{\text{dec}}^1, v_{\text{goto}}^1)$ is in $C_\mach$ if and only if
\begin{itemize}
    \item $v_{\text{goto}}^i \bmod 4 =0$  for both $i$, or if
    \item $v_{\text{goto}}^i \bmod 4 = 1$ for some $i$ (recall that $i$ is unique then) and $v_{\text{inc}}^i = v_{\text{dec}}^i$, or if
    \item $v_{\text{goto}}^i \bmod 4 = 2$ for some $i$ (again, $i$ is unique) and $v_{\text{inc}}^i \neq v_{\text{dec}}^i$.
\end{itemize}

All other requirements, e.g., Condition~\ref{casenonif} of the error definition on Page~\pageref{caseif}, the second-to-last letter not being $k-1$, and the input being in $\set{\epsilon, 0}$, can be checked using the state space of $\aut_\mach$. 
\end{proof}

\subsection{Proofs omitted in Section~\ref{sec:decision}}

First, we prove that nonemptiness for deterministic \spa is undecidable. 

\begin{proof}[Proof of Theorem \ref{thm_reachsafetyemptiness}.\ref{spaemptiness}]
The result follows immediately from Lemma \ref{lemma_safeword}:
A \pa~$(\aut, C) $ has a safe word if and only if $\Lall(\aut, C) \neq \emptyset$. 
\end{proof}

Next, we prove undecidability of nonemptiness for deterministic \cpa.

\begin{proof}[Proof of Theorem \ref{thm_reachsafetyemptiness}.\ref{cpaemptiness}]
We present a reduction from the 
universal termination problem~\cite{BlondelBKPT01}\footnote{The authors use a slightly different definition of two-counter machine than we do here. Nevertheless, the universal termination problem for their machines can be reduced to the universal termination problem for decrement-guarded two-counter machines as defined here.} for decrement-guarded two-counter machines, which is undecidable. 
The problem asks whether a given two-counter machine $\mach$ terminates from every configuration.
If this is the case, we say that $\mach$ is universally terminating.

Now, consider a decrement-guarded two-counter machine~$\mach$ that contains, without loss of generality, an increment instruction for each counter, say in lines~$\ell^+_0 $ and $\ell^+_1$.
In the  proof of Lemma~\ref{lemma_safeword}, we construct a deterministic \pa~$(\aut_\mach, C_\mach)$ that accepts a finite word~$w$ over the line numbers of $\mach$ if and only if $w = \epsilon$, $w = 0$, or if $\size{w} \geqslant 2$ and $w$ does not contain an error at position~$\size{w}-2$.
We claim that $\LcoBuchi(\aut_\mach, C_\mach)$ is nonempty if and only if $\mach$ is not universally terminating. 

So, first assume there is some $w \in \LcoBuchi(\aut_\mach, C_\mach)$. Hence, there is a run of $(\aut_\mach,C_\mach)$ processing $w$ satisfying the co-Büchi acceptance condition:
From some point~$n_0$ onward, the run only visits states in $F$ and the extended Parikh image is in $C_\mach$.
This means that there is no error in $w$ after position~$n_0$.
So, $\mach$ does not terminate from the configuration~$(w_{n_0}, c_0, c_1)$ with 
\[
c_i = \sum\nolimits_{j  \colon \instr_{j} = \inc{i}} \occ{w_0 \cdots w_{n_0-1}}{j} -
            \sum\nolimits_{j \colon \instr_{j} = \dec{i}} \occ{w_0 \cdots w_{n_0-1}}{j},
\]
i.e., $\mach$ is not universally terminating. 

Now, assume $\mach$ does not universally terminate, say it does not terminate from configuration~$(\ell, c_0, c_1)$.
Recall that the instruction in line~$\ell^+_i$ is an increment of counter~$i$.
We define $w = (\ell_0^+)^{c_0} (\ell_1^+)^{c_1} w' $ where $w'$ is the projection to the line numbers of the (nonterminating) run of $\mach$ starting in $(\ell, c_0, c_1)$.
This word does not have an error after position~$\ell$ (but may have some before that position). 
Hence there is a co-Büchi accepting run of $(\aut_\mach, C_\mach)$ processing $w$, i.e., $\LcoBuchi(\aut_\mach, C_\mach)$ is nonempty.
\end{proof}

Finally, we prove undecidability of universality for \spa.

\begin{proof}[Proof of Theorem \ref{theorem_universality}.\ref{theorem_universality_nondetsafety}]
We present a reduction from the termination problem for (decrement-guarded) two-counter machines.
So, fix such a machine~$\mach$ with line numbers~$0,1,\ldots, k-1$ where $k-1$ is the line number of the stopping instruction, fix $\Sigma = \set{0,1,\ldots,k-1}$, and consider $L_\mach = L_\mach^0 \cup L_\mach^1$ with
\begin{align*}
L_\mach^0 {}=&{} \set{ w \in \Sigma^\omega \mid  \text{$w_0 \neq 0$}} \text{ and }\\
L_\mach^1 {}=&{}\{w \in \Sigma^\omega \mid \text{if $\occ{w}{k-1}>0$ then $w$ contains an error}\\
&\hspace{5cm}\text{strictly before the first occurrence of $k-1$}\}.
\end{align*}
We first prove that $L_\mach$ is not universal if and only if $\mach$ terminates, then that $L_\mach$ is accepted by some \spa.

So, assume that $\mach$ terminates and let $w \in \Sigma^*$ be the projection of the unique finite run of $\mach$ to the line numbers.
Then, $w$ starts with $0$, contains a $k-1$, and no error before the $k-1$. 
Thus, $w0^\omega$ is not in $L_\mach$, i.e., $L_\mach$ is not universal.

Conversely, assume that $L_\mach$ is not universal.
Then, there is an infinite word~$w$ that is neither in $L_\mach^0$ nor in $L_\mach^1$.
So, $w$ must start with $0$, contain a $k-1$, but no error before the first $k-1$.
Thus, Lemma~\ref{lemma-simulation} implies that the prefix of $w$ up to and including the first $k-1$ is the projection to the line numbers of the run of $\mach$.
This run is terminating, as the prefix ends in $k-1$.
Thus, $\mach$ terminates.

It remains to argue that $L_\mach = L_\mach^0 \cup L_\mach^1$ is accepted by an \spa.
Due to closure of \spa under union and the fact that every $\omega$-regular safety property is also accepted by an \spa, we only need to consider $L_\mach^1$.
Recall that we have constructed a deterministic \pa~$(\aut, C)$ (on finite words) accepting a word if it is empty, $0$, or contains an error at the second-to-last position.
We modify this \pa into an \spa~$(\aut',C')$ that accepts $L_\mach^1$.

To this end, we add a fresh accepting state~$q_a$ and a fresh rejecting state~$q_r$ to $\aut$ while making all states of $\aut$ accepting in $\aut'$. 
Both fresh states are sinks equipped with self-loops that are labeled with $(\ell, \vec{0})$ for every line number~$\ell$. 
Here, $\vec{0}$ is the appropriate zero vector, i.e., the counters are frozen when reaching the fresh states.

Intuitively, moving to $q_a$ signifies that an error at the current position is guessed.
Consequently, if a $k-1$ is processed from a state of $\aut$, the rejecting sink~$q_s$ is reached.
We reflect in a fresh component of the extended Parikh image whether $q_a$ has been reached.
Due to this, we can define $C'$ so that the extended Parikh image of every run prefix ending in a state of $\aut$ is in $C'$ and that the extended Parikh image of a run prefix ending in $q_a$ is in $C'$ if removing the last entry (the reflecting one) yields a vector in $C$, i.e., an error has indeed occurred.
Then, we have $\Lall(\aut', C') = L_\mach^1$ as required.
\end{proof}
